\documentclass[a4paper,11pt]{article}

\usepackage{amssymb,amsmath}
\usepackage{amsthm}
\usepackage{mathtools}
\mathtoolsset{showonlyrefs=true}
\usepackage{url}

\makeatletter
\let\MyCaption\@makecaption
\makeatother
\usepackage{subcaption}
\makeatletter
\let\@makecaption\MyCaption
\makeatother
\usepackage[left=27mm,right=27mm,top=27mm,bottom=27mm]{geometry}

\newtheorem{theorem}{Theorem}[section]
\newtheorem{lemma}[theorem]{Lemma}
\newtheorem{proposition}[theorem]{Proposition}
\newtheorem{corollary}[theorem]{Corollary}
\newtheorem*{claim}{Claim}
\theoremstyle{definition}

\newtheorem{remark}[theorem]{Remark}

\title{A Cost-Scaling Algorithm for\\ Minimum-Cost Node-Capacitated Multiflow Problem}
\author{Hiroshi HIRAI and Motoki IKEDA\\
Department of Mathematical Informatics,\\
Graduate School of Information Science and Technology,\\
The University of Tokyo, Tokyo, 113-8656, Japan\\
\texttt{\normalsize \{hirai,motoki\_ikeda\}@mist.i.u-tokyo.ac.jp}}

\begin{document}

\maketitle

\begin{abstract}
In this paper, we address the minimum-cost node-capacitated multiflow problem in an undirected network.
For this problem, Babenko and Karzanov (2012) showed strongly polynomial-time solvability via the ellipsoid method.
Our result is the first combinatorial weakly polynomial-time algorithm for this problem.
Our algorithm finds a half-integral minimum-cost maximum multiflow in $O(m \log(nCD)\mathrm{SF}(kn,m,k))$ time,
where $n$ is the number of nodes, $m$ is the number of edges, $k$ is the number of terminals,
$C$ is the maximum node capacity, $D$ is the maximum edge cost,
and $\mathrm{SF}(n',m',\eta)$ is the time complexity of solving the submodular flow problem in a network of $n'$ nodes, $m'$ edges,
and a submodular function with $\eta$-time-computable exchange capacity.
Our algorithm is built on discrete convex analysis on graph structures and the concept of reducible bisubmodular flows.
\end{abstract}


\section{Introduction}
The {\em maximum free multiflow problem} and its variations/extensions
are one of well-studied subjects in combinatorial optimization,
specifically, theory of multiflows and disjoint paths; see
\cite[Part VII]{Schrijver2003Combinatorial}.
By ``free" we mean that flows can connect any pairs of terminals freely.
The problem we addressed in this paper is the
{\em minimum-cost node-capacitated free multiflow problem (MNMF)} in
an undirected network;
the formal definition will be given later.
The \emph{edge-capacitated} version of this problem,
a special case of MNMF, is classical in the literature.
Particularly, for the zero-cost case (EMF)---the problem of finding
a free multiflow of the maximum total flow-value in an
edge-capacitated undirected network,
the classical result by Lov\'{a}sz~\cite{Lovasz1976some} and
Cherkassky~\cite{Cherkasski1977solution} says that
the maximum is attained by a half-integral multiflow and can be obtained
combinatorially in strongly polynomial time.
Karzanov~\cite{Karzanov1979minimum} extended the half-integrality property
to a cost setting (MEMF).
Later, he showed in \cite{Karzanov1994Minimum} that, with help of the
ellipsoid method,
a half-integral minimum-cost maximum multiflow can be obtained in
strongly polynomial time.
Goldberg and Karzanov~\cite{Goldberg1997Scaling} gave two
`ellipsoid-free' combinatorial weakly polynomial-time algorithms for
MEMF based on capacity and cost scaling.

The {\em node-capacitated} maximum free multiflow problem (NMF),
i.e., the problem of finding a maximum free multiflow in a
node-capacitated undirected network,
was first considered by Garg, Vazirani, and
Yannakakis~\cite{Garg2004Multiway}
for approximating
the {\em node-multiway cut problem}; see \cite[Section
19.3]{Vazirani2003Approximation}.
%
Pap~\cite{Pap2007Some,Pap2008Strongly} showed the half-integrality property
and the strongly polynomial-time solvability of NMF using the ellipsoid method.
For an  `ellipsoid-free' approach,
Babenko and Karzanov~\cite{Babenko2008Scaling} developed a
combinatorial weakly polynomial-time algorithm.
For a cost setting (MNMF), Babenko and
Karzanov~\cite{Babenko2012Min}
established the half-integrality
and the strongly polynomial-time solvability via the ellipsoid method.

Recently the first author~\cite{Hirai2015L,Hirai2018Dual,Hirai2018L}
of this paper
initiated a unified approach to design efficient algorithms
for classes of multiflow and network design problems.
His approach regards the dual objective functions of multiflow problems
as discrete convex functions, {\em L-convex functions},
on certain graph structures~\cite{Hirai2018L}, and applies
techniques of {\em Discrete Convex Analysis
(DCA)}~\cite{Murota2003Discrete} previously
developed for discrete convex functions on integer lattice~$\mathbb{Z}^n$.
By this approach, called {\em DCA beyond~$\mathbb{Z}^n$},
he developed a combinatorial weakly polynomial-time algorithm for
MEMF~\cite{Hirai2015L},
and a combinatorial strongly polynomial-time algorithm for
NMF~\cite{Hirai2018Dual}.
The latter algorithm uses as a subroutine an algorithm of solving
the \emph{submodular flow problem}
with a base polyhedron defined by a small number of inequalities.
See also the survey~\cite{Hirai2017Discrete} for more details.

In this paper, we continue this line of research for designing
efficient multiflow algorithms via DCA beyond $\mathbb{Z}^n$.
Our main result is the first combinatorial
weakly polynomial-time algorithm for MNMF.

\begin{theorem}
	\label{thm:main}
	A half-integral minimum-cost maximum multiflow
	can be obtained in $O(m\cdot\log(nCD)\mathrm{SF}(kn,m, k))$ time,
	where $n$ is the number of nodes, $m$ is the number of edges, $k$ is the number of terminals,
	$C$ is the maximum node-capacity, and $D$ is the maximum edge-cost.
\end{theorem}
Here $\mathrm{SF}(n',m',\eta)$ stands for the time complexity of solving the submodular flow problem
on a network of $n'$ nodes and $m'$ edges, 
and the time complexity $\eta$
of computing the exchange capacity of the base polyhedron of the submodular function describing the problem.
By using the push-relabel algorithm by Fujishige and Zhang~\cite{Fujishige1992New},
which solves the submodular flow problem combinatorially in $O(n'^3\eta)$ time,
we can solve MNMF combinatorially in $O(mn^3k^4\log(nCD))$ time.

The proof of Theorem~\ref{thm:main} is outlined as follows.
Our algorithm is designed largely on the basis of the approach in \cite{Hirai2018Dual}.
By sharpening the LP-duality,
the dual of MNMF is formulated as an optimization problem
over a grid structure $\mathbb{G}^n$,
where $\mathbb{G}$ is an amalgamation of $k$ planar grid-graphs as in Figure~\ref{fig:twisted_grid}.
The objective function is shown to be an L-convex function on
$\mathbb{G}^n$.
An L-convex function can be minimized by a generic descent algorithm, called
the {\em steepest descent algorithm (SDA)}.
An abstract description of SDA is quite simple:
Find a ``local" minimizer (called a {\em steepest descent
direction}) around the current point,
update the current point if the objective value decreases,
and repeat until the process gives no improvement.
For NMF, \cite{Hirai2018Dual} showed that SDA can be efficiently implemented by
using a submodular flow algorithm
and brings a (strongly) polynomial-time algorithm.
In extending this approach to MNMF, there are two technical issues
that should be resolved:
\begin{itemize}
\item How do we find a steepest descent direction efficiently?
Additionally, even if an optimal dual solution is obtained by SDA,
how do we recover an optimal multiflow?
We here need to deal with a more general L-convex function than that
in \cite{Hirai2018Dual}.
\item Even if a steepest descent direction is available in each step,
the number of iterations of SDA depends on the edge costs;
therefore SDA gives only a pseudopolynomial-time algorithm.
\end{itemize}
The technical contribution in showing Theorem~\ref{thm:main}
is to resolve the two issues.

\begin{figure}[t]
\centering
\includegraphics[scale=0.6]{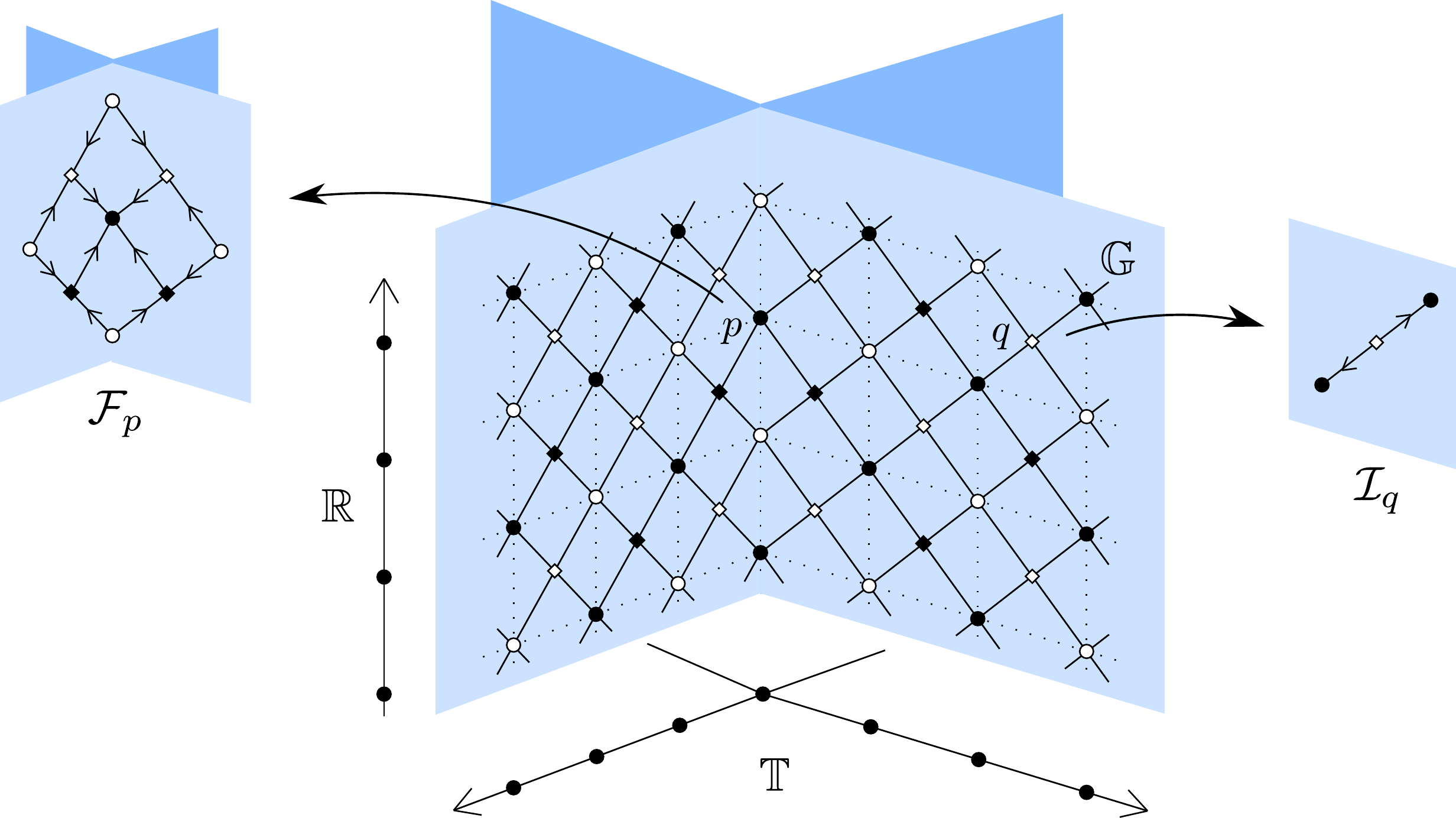}
\caption{Grid $\mathbb{G}\subseteq \mathbb{T}\times\mathbb{R}$ given in Section~\ref{subsec:l_conv} ($k=4$).}
\label{fig:twisted_grid}
\end{figure}
 
For the first issue, as was suggested in~\cite{Hirai2018Dual},
we show that (a) checking the optimality of a dual solution,
(b) finding a steepest descent direction, and
(c) recovery of an optimal multiflow
are reduced to the {\em bisubmodular flow (feasibility) problem}---the problem of
finding a (fractional) bidirected flow whose boundary belongs
to a given bisubmodular polyhedron.
Although the bisubmodular flow problem itself is a natural problem,
it has not been well-studied so far.
Particularly, we need to deal with the {\em node-flowing polytope},
which is a bisubmodular polyhedron that represents
the node-capacity constraint and the flow-conservation law on a node.
\cite{Hirai2018Dual} showed that the node-flowing polytope for a node
of degree 3
is represented as the projection of the base polytope of a submodular function
on a 6-element set.
By this property and a perturbation technique specific to NMF,
\cite{Hirai2018Dual} reduced the bisubmodular flow problem for NMF
to a submodular flow (feasibility) problem on a larger set.
In this paper, we prove in Lemma~\ref{lem:node_flowing} that such a lifting property holds
for the general node-flowing polytope, and
reduce the bisubmodular flow problem for MNMF to
a submodular flow problem (with a small number of inequalities).
For dealing with such a bisubmodular flow problem systematically,
we develop a mini-theory for
the {\em reducible bisubmodular flow problem},
which may be of independent interest.

To overcome the second issue, we combine the cost-scaling with SDA for
improving the pseudopolynomial-time algorithm.
Though the cost-scaling is a standard idea in network flow theory,
for complexity analysis it requires the estimate
of the number of iterations in each scaling phase.
It is known \cite{Hirai2018Dual} (see also \cite{Hirai2017Discrete,Hirai2018L})
that the number of iterations of SDA is bounded by
the distance between the initial point and minimizers.
We prove a sensitivity theorem (Theorem~\ref{thm:sensitivity})
that minimizers of our L-convex function do not vary so much for a small change of an edge cost.
Combining them, we obtain a polynomial bound of iterations of the
cost-scaling SDA, and
complete the proof of Theorem~\ref{thm:main}.

The rest of this paper is organized as follows.
In Section~\ref{sec:preliminaries},
we present preliminary results on
discrete convexity and (reducible) bisubmodular flow.
In Section~\ref{sec:mincost},
we study our multiflow problem MNMF and give the proof of
Theorem~\ref{thm:main}.
Some of proofs are technical and proved in Section~\ref{sec:proofs}.


\section{Preliminaries}
\label{sec:preliminaries}

Let $\mathbb{R}$, $\mathbb{R}_+$, $\mathbb{Z}$, and $\mathbb{Z}_+$
denote the sets of reals, nonnegative reals, integers, and nonnegative integers.
Let $\overline{\mathbb{R}}$ and $\underline{\mathbb{R}}$ denote $\mathbb{R}\cup \{+\infty\}$ and $\mathbb{R}\cup \{-\infty\}$, respectively.
For $x\in \mathbb{R}$, let $\lceil x\rceil$ denote the minimum integer which is not smaller than $x$.
For a finite set $V$, a subset $X\subseteq V$, and a function $g:V\rightarrow \overline{\mathbb{R}}$,
let $g(X)$ denote the sum $\sum_{i\in X} g(i)$.
We sometimes denote the function value $g(i)$ by $g_i$ if no confusion occurs.
For $i\in V$, let $\chi_i:V\rightarrow \mathbb{R}$ be a function defined by
\[
    \chi_i(j):=\begin{dcases*} 1 & if $i=j$,\\ 0 & otherwise\end{dcases*}
    \quad (j\in V).
\]
For $X,Y\in V$ with $X\cap Y=\emptyset$, let $\chi_{X,Y}:=\sum_{i\in X}\chi_i-\sum_{j\in X}\chi_j$.


\subsection{L-convex functions on $\mathbb{G}^n$}
\label{subsec:l_conv}
In this subsection, we briefly introduce a class of discrete convex functions (L-convex functions)
on a certain grid-like graph structure $\mathbb{G}^n$.
Then we describe the steepest descent algorithm for minimization of these functions.
See \cite{Hirai2017Discrete,Hirai2018L} for further details.

To begin with, let us recall the notion of L$^\natural$-convex functions on $\mathbb{Z}^n$, 
where our L-convex functions will be analogously introduced.
A function $h:\mathbb{Z}^n \to \overline{\mathbb{R}}$ is an \emph{L$^\natural$-convex function} 
if it satisfies the discrete midpoint convexity inequality
\begin{equation}\label{eq:dci}
h(p) + h(q) \geq h(\lceil (p + q)/2 \rceil) + h(\lfloor (p + q)/2 \rfloor) \quad (p,q \in \mathbb{Z}^n),
\end{equation}
where  
$\lceil \cdot \rceil$ and $\lfloor \cdot \rfloor$ are 
the rounding-up and -down operators 
that rounds up and down the fractional part of each component, respectively.
See \cite{Murota2003Discrete} for details of L$^\natural$-convex functions.

We are going to introduce a continuous space 
$\mathbb{T} \times \mathbb{R}$ and a discrete subspace $\mathbb{G} \subseteq \mathbb{T} \times \mathbb{R}$,
in which $\mathbb{T} \times \mathbb{R}$ 
has a midpoint operation $(p,q) \mapsto (p+q)/2$ and $\mathbb{G}$ 
has discrete midpoint operations $(p,q) \mapsto \lceil (p + q)/2 \rceil$ and $(p,q) \mapsto \lfloor (p + q)/2 \rfloor$.
Then an L-convex function on $\mathbb{G}^n$ is defined by the discrete midpoint convexity inequality \eqref{eq:dci}
with replacing $\mathbb{Z}^n$ by $\mathbb{G}^n$. 

Let us first define the notion of the infinite $k$-star $\mathbb{T}$, where $k$ is a positive integer.
Informally speaking, the infinite $k$-star 
is a ``continuous star" obtained by gluing $k$ semi-infinite paths at the common endpoint;
see the bottom of Figure~\ref{fig:twisted_grid}.
Consider the set $\mathbb{R}_+\times\{1,2,\dotsc,k\}$ of $k$ copies of $\mathbb{R}_+$,
and define equivalence relation $\sim$ on it
by $(x,s) \sim (x',s')$ if $(x,s)=(x',s')$ or $x = x' = 0$.
Then the quotient space $\mathbb{T} := \mathbb{R}_+ \times \{1,2,\ldots,k\} /\sim$
is called the {\em infinite $k$-star}.
For notational simplicity, the point $(x,s) \in \mathbb{T}$ is also denoted by $x$.
For two points $x = (x,s), x' = (x',s') \in \mathbb{T}$, 
we denote the shortest distance on $\mathbb{T}$ by $\mathrm{dist}(x,x')$, i.e., 
$\mathrm{dist}(x,x')=|x - x'|$ if $s=s'$ and $\mathrm{dist}(x,x')=x+x'$ if $s \neq s'$.
Then $\mathbb{T}$ becomes a metric space with the distance function $\mathrm{dist}(\cdot,\cdot)$.
The {\em midpoint} of two points $x,x' \in \mathbb{T}$ 
is defined as the unique point $u \in \mathbb{T}$ such that $\mathrm{dist}(x,u) = \mathrm{dist}(u,x')$
and $\mathrm{dist}(x,x') = \mathrm{dist}(x,u) + \mathrm{dist}(u,x')$.
The midpoint $u$ is denoted symbolically by $(x+x')/2$.
A point $x =(x,s)$ is called {\em integral} if $\mathrm{dist}(x,0) \in \mathbb{Z}$ and {\em half-integral} 
if $2\mathrm{dist}(x,0) \in \mathbb{Z}$.
Also $x$ is called \emph{proper half-integral} if $x$ is half-integral but not integral.


Second, consider the product $\mathbb{T} \times \mathbb{R}$ of the infinite $k$-star $\mathbb{T}$ 
and the set $\mathbb{R}$ of reals.
Let $\mathbb{G} \subseteq \mathbb{T} \times \mathbb{R}$ be the discrete subset of $\mathbb{T}$ consisting of 
points $(x,y)$ such that both $x$ and $y$ are integral or both $x$ and $y$ are proper half-integral; see Figure~\ref{fig:twisted_grid}.
A point $p=(x,y) \in \mathbb{G}$ is called {\em integral} 
if both $x$ and $y$ are integral.
An integral point $p =(x,y)$ is said to be {\em even} (resp. \emph{odd}) 
if $\lvert x-y\rvert$ is even (resp. \emph{odd}).
In Figure~\ref{fig:twisted_grid}, even and odd points are drawn in black and white circles, respectively,
and black and white diamonds are non-integral points (a meaning of color is described in Section~\ref{subsec:sdd}).
For two points $p=(x,y),q=(x',y')\in\mathbb{T}\times\mathbb{R}$, 
let $\lVert p-q\rVert$ denote $\mathrm{dist}(x,x') + |y-y'|$.
Note that $\lVert p-q\rVert$ is integral for any $p,q\in\mathbb{G}$.
Two points $p,q \in \mathbb{G}$ are said to be {\em adjacent} 
if $\|p - q\| = 1$ and exactly one of $p$ and $q$ is integral.
By joining each of adjacent points by an edge 
the set $\mathbb{G}$ is viewed as a grid on $\mathbb{T}\times\mathbb{R}$.

\begin{figure}[t]
\centering
\includegraphics[scale=0.8]{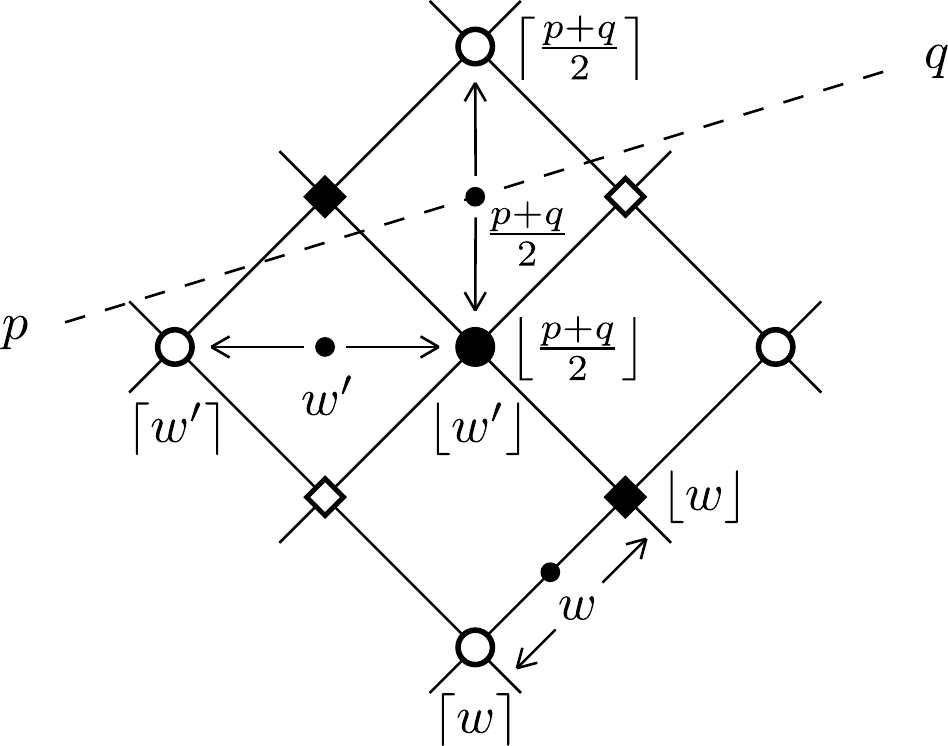}
\caption{Midpoint operation: Here $w=(p+q)/2$ is not necessarily a point in $\mathbb{G}$.
If $w\notin \mathbb{G}$, then $w$ is the midpoint of an edge $uv$ of $\mathbb{G}$
or the center of a square (4-circuit) of $\mathbb{G}$.
For the former case, $(\lfloor w\rfloor,\lceil w\rceil)=(u,v)$ if $u\prec v$.
For the latter case, $\lfloor w\rfloor$ and $\lceil w\rceil$ are the even and odd points, respectively, of the square.}
\label{fig:mid}
\end{figure}

Third, we define discrete midpoint operations on $\mathbb{G}$.
For two points $p = (x,y), q = (x',y')\in \mathbb{G}$, the midpoint 
$(p+q)/2 := ((x+x')/2, (y+y')/2) \in \mathbb{T} \times \mathbb{R}$ 
is not necessarily a point in $\mathbb{G}$.   
As in the case of $\mathbb{Z}^n$, we consider rounding-up and -down 
of this point, which is defined via a partial order $\preceq$ on $\mathbb{G}$.
For each pair $p,q$ of adjacent points in $\mathbb{G}$,
$p \prec q$ is defined to hold if $p$ is even or $q$ is odd.
Then $\preceq$ is defined as the reflexive transitive closure of $\prec$.
Now odd and even points are maximal and minimal elements in the poset $\mathbb{G}$, respectively.
For two point $p,q \in \mathbb{G}$, there is a unique pair of points $u,v \in \mathbb{G}$ such that 
$u \preceq v$ and $(p+q)/2 = (u+v)/2$.
Then $u$ and $v$ are denoted by $\lfloor (p + q)/2 \rfloor$ and $\lceil (p + q)/2 \rceil$; see Figure~\ref{fig:mid}.

Finally, we can define an L-convex function on $\mathbb{G}^n$.
A function $h:\mathbb{G}^n\rightarrow \overline{\mathbb{R}}$ is an \emph{L-convex function}
if it satisfies
\[
h(p)+h(q)\geq h(\lfloor (p+q)/2\rfloor)+h(\lceil (p+q)/2\rceil) \quad (p,q\in \mathbb{G}^n).
\]

It is known that an L-convex function $h$ can be minimized
by the following algorithm, which is called the \emph{steepest descent algorithm (SDA)}:
\begin{enumerate}
\setcounter{enumi}{-1}
\item Find an initial point $p\in \mathbb{G}^n$ with $h(p)<\infty$.
\item Let
\begin{equation}
\label{eq:neighbor}
\begin{aligned}
\mathcal{F}_p&:=\{q\in \mathbb{G}^n\mid q_i\succeq p_i\ (i=1,\dotsc,n)\},\\
\mathcal{I}_p&:=\{q\in \mathbb{G}^n\mid q_i\preceq p_i\ (i=1,\dotsc,n)\}.
\end{aligned}
\end{equation}
\item Find a minimizer $q$ of $h$ over $\mathcal{F}_p\cup\mathcal{I}_p$.
\item If $h(q)<h(p)$ then $q$ is called a \emph{steepest descent direction} at $p$; update $p$ by $q$ and back to step 1.
    Otherwise stop; return $p$.
\end{enumerate}
The left part of Figure~\ref{fig:twisted_grid} shows $\mathcal{F}_p$ of an even integral point $p$,
and the right part shows $\mathcal{I}_q$ of a non-integral point $q$.
If SDA terminates, the output $p^*$ is a minimizer of $h$.

\begin{lemma}[{\cite[Lemma 2.3]{Hirai2017Discrete}}]
\label{lem:SDAopt}
If $p\in \mathbb{G}^n$ with $h(p)<\infty$ is not a minimizer of an L-convex function $h$,
then there exists $q\in \mathcal{F}_p\cup\mathcal{I}_p$ such that $h(q)<h(p)$.
\end{lemma}

%
%
%

Define a kind of $\ell_\infty$-distance on $\mathbb{G}^n$ by
\begin{equation}
\label{eq:dist_on_grid}
    \lVert p - q \rVert:=\max_{1\leq i\leq n} \lVert p_i-q_i \rVert\quad (p,q\in \mathbb{G}^n).
\end{equation}
The number of iterations of SDA is bounded by $\lVert\cdot\rVert$ as below.

\begin{theorem}[{\cite[Theorem 4.3]{Hirai2018L}}]
\label{thm:SDA}
Let $\mathrm{opt}(h)$ be the set of minimizers of $h$, and let $p$ be the initial point of SDA.
If $\mathrm{opt}(h)\neq\emptyset$, then SDA terminates
after at most $\min_{q\in \mathrm{opt}(h)}\lVert p-q\rVert +2$ iterations.
\end{theorem}

Currently the polynomial-time solvability of Step 2 of SDA is unknown.
Thus it does not imply the polynomial-time solvability of the L-convex function minimization (under the value oracle model).

\subsection{Submodular and bisubmodular flows}
\label{subsec:subbisub}

\subsubsection{Submodular flow}

For a finite set $U$, a function $\rho : 2^U\rightarrow \overline{\mathbb{R}}$ with $\rho(\emptyset)=0$ is called \emph{submodular}
if it satisfies $\rho(X)+\rho(Y)\geq \rho(X\cap Y) + \rho(X\cup Y)$ for all $X,Y\subseteq U$.
For a submodular function $\rho$,
the \emph{base polyhedron} $\mathcal{B}(\rho)$ is
the set of all vectors $x\in \mathbb{R}^U$ satisfying
\begin{gather}
x(U)=\rho(U),\\
x(X)\leq \rho(X)\quad (X\subseteq U).\label{eq:base}
\end{gather}

Let $(U,A)$ be a directed graph on node set $U$.
Let $\underline{b}:A\rightarrow\underline{\mathbb{R}}$ and $\overline{b}:A\rightarrow\overline{\mathbb{R}}$
be lower and upper capacity functions, respectively.
Assume that $\underline{b}(a)\leq \overline{b}(a)$ for each $a\in A$.
Let $\rho:U\rightarrow \overline{\mathbb{R}}$ be a submodular function with $\rho(U)=0$.
For a function (flow) $\varphi:A\rightarrow\mathbb{R}$, let $\nabla\varphi\in \mathbb{R}^U$ be the \emph{boundary} of $\varphi$ defined by
\[
(\nabla \varphi)(u):=\sum\{\varphi(a)\mid \text{$a$ leaves $u$}\}-\sum\{\varphi(a)\mid \text{$a$ enters $u$}\}\quad (u\in U).
\]
The \emph{submodular flow (feasibility) problem (SF)} is the problem of finding a flow $\varphi:A\rightarrow \mathbb{R}$ satisfying 
\begin{gather}
        \nabla \varphi \in \mathcal{B}(\rho),\\
        \underline{b}(a) \leq \varphi(a)\leq \overline{b}(a)\quad (a\in A).
\end{gather}
We say that a flow $\varphi$ is \emph{feasible} if it satisfies the above conditions.
The feasibility of SF can be characterized via its \emph{cut function} $\kappa_{\underline{b},\overline{b}}$ defined by
\begin{align*}
    \kappa_{\underline{b},\overline{b}}(X)
        &:=\sum\{\underline{b}(a)\mid a=ij\in A, i\in X, j\notin X\}\\
        &\qquad-\sum\{\overline{b}(a)\mid a=ij\in A, i\notin X, j\in X\}\quad (X\subseteq U).
\end{align*}
It is well-known that $-\kappa_{\underline{b},\overline{b}}(X)$ is submodular; see, e.g., \cite[Section 2.3]{Fujishige2005Submodular}.
We denote the cut function simply by $\kappa(X)$ if no confusion occurs.

\begin{theorem}[\cite{Frank1984Finding}]
\label{thm:submod_feasibility}
SF has a feasible flow if and only if
it holds $\kappa(X)\leq \rho(X)$ for all $X\subseteq U$.
In such a case, if $\underline{b},\overline{b}$ and $\rho$ are all integer-valued,
there exists an integral feasible flow.
\end{theorem}

We say that $X\subseteq U$ with $\kappa(X)>\rho(X)$ is a \emph{violating cut}.
A violating cut $X$ is called \emph{maximum}
if $\kappa(X)-\rho(X)$ is maximum among all violating cuts.


There are a number of algorithmic results on SF; see the survey \cite{Fujishige2000Algorithms}.
Those algorithms output a feasible flow when SF is feasible, and output a violating cut when SF is infeasible.
If $\underline{b}$, $\overline{b}$, and $\rho$ are all integer-valued,
the above flow can be chosen as integral.
Moreover, by modifying the algorithms like \cite[Section 16.2.4]{Frank2011Connections},
the above cut can be chosen as a maximum violating cut.
For designing algorithms, one usually assume the \emph{exchange capacity oracle} for $\rho$.
For $x\in \mathcal{B}(\rho)$ and a distinct pair $(i,j)$ of $U$,
the \emph{exchange capacity} is the minimum value of $\rho(X)-x(X)$ among all $X\subseteq U$ with $i\in X \not\ni j$.

\subsubsection{Bisubmodular flow}
\label{subsubsec:bisub}

Let $3^U$ denote the set of all pairs $Y,Z\subseteq U$
satisfying $Y\cap Z=\emptyset$.
A function $\beta : 3^U\rightarrow \overline{\mathbb{R}}$ with $\beta(\emptyset,\emptyset)=0$ is called \emph{bisubmodular}
if it satisfies the bisubmodularity inequality
\begin{align}
    \beta(Y,Z)+\beta(Y',Z')&\geq \beta(Y\cap Y',Z\cap Z')\\
        &\quad +\beta((Y\cup Y')\setminus (Z\cup Z'),(Z\cup Z')\setminus (Y\cup Y'))
\end{align}
for all $(Y,Z),(Y',Z')\in 3^U$.
For a bisubmodular function $\beta$ on $U$, the \emph{bisubmodular polyhedron} $\mathcal{D}(\beta)$ is
the set of all vectors $z\in \mathbb{R}^U$ satisfying
\begin{equation}
\label{eq:bisub_poly}
z(Y)-z(Z)\leq \beta(Y,Z)\quad ((Y,Z)\in 3^U).
\end{equation}

A \emph{bidirected graph} is a graph $(U,E;\partial)$,
where $\partial:E\rightarrow \mathbb{Z}^U$ is a \emph{boundary operator} such that
for each $ij\in E$ it satisfies (i) $\partial e=\chi_i+\chi_j$, (ii) $\partial e=-\chi_i-\chi_j$,
(iii) $\partial e=\chi_i-\chi_j$, or (iv) $\partial e=-\chi_i+\chi_j$.
An undirected graph can be identified with a bidirected graph whose all arcs have boundaries of type (i),
and a directed graph can be identified with a bidirected graph whose all arcs have boundaries of type (iii).
In this paper, we do not admit self-loops of type (iii) or (iv).

Let $(U,E;\partial)$ be a bidirected graph.
Let $\underline{c}:E\rightarrow \underline{\mathbb{R}}$ and $\overline{c}:E\rightarrow \overline{\mathbb{R}}$
be lower and upper capacities, respectively.
Assume that $\underline{c}(e)\leq \overline{c}(e)$ for each $e\in E$.
Let $\beta:3^U\rightarrow \overline{\mathbb{R}}$ be a bisubmodular function.
For a function (bidirected flow) $\psi:E\rightarrow\mathbb{R}$, let $\nabla\psi \in \mathbb{R}^U$ be the \emph{boundary} of $\psi$
defined by $\nabla\psi:=\sum\{\psi(e)\partial e\mid e\in E\}$.
Note that this definition is consistent with that given in a directed case.
The \emph{bisubmodular flow (feasibility) problem (BF)} is the problem of finding a flow $\psi:E\rightarrow\mathbb{R}$ satisfying
\begin{equation}
\label{eq:bisub_prob}
\begin{gathered}
    \nabla \psi \in \mathcal{D}(\beta),\\
    \underline{c}(e) \leq \psi(e)\leq \overline{c}(e)\quad (e\in E).
\end{gathered}
\end{equation}
We say that a flow $\psi$ is \emph{feasible} if it satisfies the above conditions.

There are few papers~\cite{Karzanov1991Maximization,Karzanov20070} on BF despite its natural formulation.
We give a similar characterization of feasibility of BF as SF.
Let $\hat{\kappa}_{\underline{c},\overline{c}}:3^U\rightarrow \overline{\mathbb{R}}$ be the \emph{cut function} of the network defined by
\begin{align*}
    \hat{\kappa}_{\underline{c},\overline{c}}(Y,Z)
        &:=\sum\{\langle \partial e,\chi_{Y,Z}\rangle \underline{c}(e)\mid e\in E,\langle \partial e,\chi_{Y,Z}\rangle >0\}\\
        &\qquad+\sum\{\langle \partial e,\chi_{Y,Z}\rangle \overline{c}(e)\mid e\in E,\langle \partial e,\chi_{Y,Z}\rangle <0\}\quad ((Y,Z)\in 3^U),
\end{align*}
where $\langle\cdot,\cdot\rangle$ is the canonical inner product.
We denote the cut function simply by $\kappa(Y,Z)$ if no confusion occurs.

\begin{theorem}
\label{thm:bisubmod_feasibility}
BF has a feasible flow if and only if
it holds $\kappa(Y,Z)\leq \beta(Y,Z)$ for all $(Y,Z)\in 3^U$.
In such a case, if $\underline{c},\overline{c}$ and $\beta$ are all integer-valued,
there exists a half-integral feasible flow.
\end{theorem}

Again, we say that $(Y,Z)\in 3^U$ with $\kappa(Y,Z)>\beta(Y,Z)$ is a \emph{violating cut}.

\begin{proof}
These are known~\cite{Ando1997Balanced} that $\beta'(Y,Z):=-\kappa(Z,Y)$ is bisubmodular and
\[
	\mathcal{D}(\beta')=\{\nabla \psi\mid \psi\in \mathbb{R}^E,
	\underline{c}\leq \psi\leq\overline{c}\}.
\]
Therefore, BF is to find a vector which belongs to the intersection of two bisubmodular polyhedra.
It is shown~\cite{Nakamura1990Intersection} that the intersection $\mathcal{D}(\beta)\cap \mathcal{D}(\beta')$
of two bisubmodular polyhedra is nonempty if and only if
$\beta(Y,Z)+\beta'(Z,Y)\geq 0$ for any $(Y,Z)\in 3^U$.
Then the former part of the theorem immediately follows.

It is a folklore that the intersection of two integral bisubmodular polyhedra is half-integral;
see, e.g., \cite{Karzanov20070}.
This implies that if BF has a feasible flow and $\underline{c}$, $\overline{c}$, and $\beta$ are integral-valued,
one can choose such a flow as its boundary is half-integral.
For showing the half-integrality of a flow itself, divide each bidirected edge $e=ij\in E$ with $\partial e=\sigma_i\chi_i+\sigma_j\chi_j$
into two bidirected edges $ik_{e}$ and $k'_{e}j$,
where $\sigma_i,\sigma_j\in\{+1,-1\}$, and $k_{e}$ and $k'_{e}$ are new vertices.
Define their boundaries by $\partial (ik_{e})=\sigma_i\chi_i+\chi_{k_{e}}$ and $\partial (k'_{e}j)=\chi_{k'_{e}}+\sigma_j\chi_j$,
and their capacities by $\underline{c}(ik_e):=\underline{c}(k'_e j):=\underline{c}(e)$
and $\overline{c}(ik_e):=\overline{c}(k'_e j):=\overline{c}(e)$.
Let $\tilde{E}$ be the set of divided edges.
Consider the constraint $(\nabla \psi)(k_e)=(\nabla \psi)(k'_e)$
(or, equivalently, $\psi(ik_e)=\psi(k'_e j)$).
It is easy to see that this is the boundary condition for
the following bisubmodular function $\beta_e$ on $\{k_e,k'_e\}$:
\[
    \beta_e(Y,Z):=\begin{dcases*}
        0 & if $(Y,Z)=(\emptyset,\emptyset),(\{k_e\},\{k'_e\}),(\{k'_e\},\{k_e\})$,\\
        \infty & otherwise
    \end{dcases*}\quad ((Y,Z)\in 3^{\{k_e,k'_e\}}).
\]
Thus, extend the bisubmodular function $\beta$ to $\tilde{\beta}:3^{\tilde{U}}\rightarrow\overline{\mathbb{R}}$
with $\tilde{U}=U\cup \bigcup_{e\in E}\{k_e,k'_e\}$ by
\[
    \tilde{\beta}(\tilde{Y},\tilde{Z}):=\beta(\tilde{Y}\cap U,\tilde{Z}\cap U)+\sum_{e\in E}\beta_e(\tilde{Y}\cap \{k_e,k'_e\},\tilde{Z}\cap \{k_e,k'_e\})
    \quad ((\tilde{Y},\tilde{Z})\in 3^{\tilde{U}}).
\]
The new instance $((\tilde{U},\tilde{E};\partial),\underline{c},\overline{c},\tilde{\beta})$
of BF is obviously equivalent to the original.
Now the half-integrality of $\psi$ follows from the half-integrality of $\nabla \tilde{\psi}$ (and thus, $\tilde{\psi}$).
\end{proof}

\subsubsection{Reducible bisubmodular flow}
\label{subsec:reducible}


The \emph{signed extension} $U^\pm$ of $U$ is defined by $U^\pm:=U\times\{-1,+1\}$,
where elements $(i,+1)$ and $(i,-1)$ in $U^\pm$ are denoted by $i^+$ and $i^-$, respectively.
For $X\subseteq U^\pm$, let $\underline{X}$ be the set obtained from $X$
by removals of all $\{i^+,i^-\}$ with $\{i^+,i^-\}\subseteq X$.
For $Y\subseteq U$, let $Y^+:=\{i^+\mid i\in Y\}\subseteq U^\pm$ and $Y^-:=\{i^- \mid i \in Y\}\subseteq U^\pm$.

A function $\rho:2^{U^\pm}\rightarrow\overline{\mathbb{R}}$ is called \emph{transversally monotone} (\emph{t-monotone}) if
\[
    \rho(\underline{X})\leq \rho(X)\quad (X\subseteq 2^{U^\pm}).
\]
Define a map (projection) $\Phi:\mathbb{R}^{U^\pm}\rightarrow \mathbb{R}^U$ and a map (lift) $\Psi:\mathbb{R}^U\rightarrow\mathbb{R}^{U^\pm}$ by
\[
    (\Phi(x))(i):=\frac{x(i^+)-x(i^-)}{2},\quad (\Psi(z))(i^+):=z(i),\ (\Psi(z))(i^-):=-z(i)
\]
for $x\in\mathbb{R}^{U^\pm}$, $z\in\mathbb{R}^U$, and $i\in U$.
The following is a variant of \cite[Lemma 2.3]{Hirai2018Dual}.

\begin{lemma}
\label{lem:reducible}
Let $\rho:2^{U^\pm}\rightarrow\overline{\mathbb{R}}$ be a t-monotone submodular function.
If $\Psi(\Phi(\mathcal{B}(\rho)))\subseteq \mathcal{B}(\rho)$,
then $\Phi(\mathcal{B}(\rho))$ is the bisubmodular polyhedron $\mathcal{D}(\beta)$
of a bisubmodular function $\beta:3^U\rightarrow \overline{\mathbb{R}}$ defined by
\begin{equation}
\label{eq:reducible_bisub}
    \beta(Y,Z):=\rho(Y^+\cup Z^-)\quad ((Y,Z)\in 3^U).
\end{equation}
\end{lemma}

\begin{proof}
First we check the bisubmodularity of $\beta$:
For $(Y,Z),(Y',Z')\in 3^U$, we have
\begin{align}
    &\beta(Y,Z)+\beta(Y',Z') = \rho(Y^+\cup Z^-)+\rho(Y'^+\cup Z'^-)\\
        &\quad\geq \rho((Y\cap Y')^+\cup (Z\cap Z')^-) + \rho((Y\cup Y')^+\cup (Z\cup Z')^-)\\
        &\quad \geq\beta(Y\cap Y',Z\cap Z') +\beta((Y\cup Y')\setminus (Z\cup Z'),(Z\cup Z')\setminus (Y\cup Y')),
\end{align}
where the second inequality follows from the t-monotonicity.
Next we show $\Psi(\mathcal{B}(\rho))=\mathcal{D}(\beta)$.
Suppose that $z\in \Phi(\mathcal{B}(\rho))$.
Since $\Psi(z)\in \mathcal{B}(\rho)$, we have
\[
    z(Y)-z(Z)=(\Psi(z))(Y^+\cup Z^-)\leq \rho(Y^+\cup Z^-)=\beta(Y,Z)\quad ((Y,Z)\in 3^U).
\]
Thus $z\in \mathcal{D}(\beta)$.
Conversely, suppose that $z\in \mathcal{D}(\beta)$.
Let $X\subseteq U^\pm$ and take $Y,Z\subseteq U$ so that $\underline{X}=Y^+\cup Z^-$. Since $Y\cap Z=\emptyset$, we have
\[
    \Psi(z)(X)=\Psi(z)(\underline{X})=z(Y)-z(Z)\leq \beta(Y,Z)=\rho(\underline{X})\leq \rho(X).
\]
Hence $\Psi(z)\in \mathcal{B}(\rho)$, implying $z=\Phi(\Psi(z))\in \Phi(\mathcal{B}(\rho))$.
\end{proof}

By the fact that the bisubmodular function is uniquely determined
for a bisubmodular polyhedron (see, e.g., \cite[Section 3.5(b)]{Fujishige2005Submodular}),
$\beta$ in the above lemma is unique.
A bisubmodular function $\beta$ is called \emph{reducible} (or \emph{reducible to $\rho$})
if it has a submodular function $\rho$ satisfying the assumptions of the above lemma.
Also, BF is called \emph{reducible} if its bisubmodular function is reducible.
Since $x(U^\pm)=0$ for any $x\in \Psi(\Phi(\mathcal{B}(\rho)))\subseteq \mathcal{B}(\rho)$,
we observe that $\rho(U^\pm)=0$.
Thus we can consider an instance of SF on $U^\pm$ using $\rho$. 
We show that the reducible BF can be reduced to the SF.

We construct a directed graph $(U^\pm,A)$ as follows.
For each $e\in E$ with $\partial e=\sigma_i\chi_i+\sigma_j\chi_j$,
let $A_e:=\{i^{\sigma_i}j^{-\sigma_j},j^{\sigma_j}i^{-\sigma_i}\}$.
It is allowed to form parallel arcs from a self-loop.
The arc set $A$ is defined to be the disjoint union of all $A_e$.
Lower and upper capacity functions on $A$ are defined by
$\underline{b}(a)=\underline{c}(e)$ and $\overline{b}(a)=\overline{c}(e)$ for $a\in A_e$, respectively.

A violating cut $(Y,Z)\in 3^U$ is \emph{maximum}
if $\hat{\kappa}_{\underline{c},\overline{c}}(Y,Z)-\beta(Y,Z)$ is maximum among all violating cuts.

\begin{proposition}
\label{prop:reduction_bisub}
Suppose that $\beta$ is reducible to $\rho$.
Let $((U,E;\partial),\underline{c},\overline{c},\beta)$ be an instance of BF,
and $((U^\pm,A),\underline{b},\overline{b},\rho)$ be an instance of SF.

\begin{enumerate}
\renewcommand{\labelenumi}{{\rm (\arabic{enumi}).}}
\item  BF is feasible if and only if SF is feasible.
\item Suppose that BF (or SF) is infeasible.
    Then $(Y,Z)\in 3^U$ is a maximum violating cut for BF
    if and only if $Y^+\cup Z^-$ is a maximum violating cut for SF.
\end{enumerate}
\end{proposition}

\begin{proof}
(1) Suppose that SF is feasible.
Let $\varphi:A\rightarrow\mathbb{R}$ be a feasible flow for SF.
Define a flow $\psi:E\rightarrow\mathbb{R}$ for BF by
\begin{equation}
\label{eq:give_bisubflow_by_subflow}
    \psi(e)=\frac{1}{2}\,\varphi(A_e)\quad (e\in E).
\end{equation}
Then, it is clear that $\underline{c}(e) \leq \psi(e)\leq \overline{c}(e)$ for each $e\in E$.
Moreover, for $i\in V$, we can check from the definition that
\begin{equation}
\label{eq:relation_nabla}
    (\nabla \psi)(i) = \frac{(\nabla \varphi)(i^+)-(\nabla\varphi)(i^-)}{2}.
\end{equation}
Thus we have $\nabla \psi = \Phi(\nabla \varphi)\in \Phi(\mathcal{B}(\rho))=\mathcal{D}(\beta)$.
Therefore, $\psi$ is a feasible flow for BF.


Suppose that SF is infeasible.
These are shown~\cite{Ando1997Balanced}\footnote{Actually, in \cite{Ando1997Balanced}, these facts are shown only for the case $\underline{b}=0$.
But the same proofs are running in the case $\underline{b}\neq 0$.} that
\begin{align}
\label{eq:cut_reduce} \hat{\kappa}_{\underline{c},\overline{c}}(Y,Z)&=\kappa_{\underline{b},\overline{b}}(Y^+\cup Z^-)\quad ((Y,Z)\in 3^U),\\
\label{eq:cut_monotone} \kappa_{\underline{b},\overline{b}}(X)&\leq \kappa_{\underline{b},\overline{b}}(\underline{X})\quad (X\subseteq U^\pm).
\end{align}
Let $X\subseteq U^\pm$ be a violating cut for SF,
and take $Y,Z\subseteq U$ so that $\underline{X}=Y^+\cup Z^-$.
Then it follows from the t-monotonicity of $\rho$ that
\begin{equation}
\label{eq:violating_bisub}
    \hat{\kappa}_{\underline{c},\overline{c}}(Y,Z)=\kappa_{\underline{b},\overline{b}}(\underline{X})
    \geq \kappa_{\underline{b},\overline{b}}(X)>\rho(X)\geq \rho(\underline{X})=\beta(Y,Z).
\end{equation}
Thus, $(Y,Z)$ is a violating cut for BF.

\vspace{5pt}
\noindent (2) Suppose that $(Y,Z)$ is a maximum violating cut for BF,
and $Y^+\cup Z^-$ is not a maximum violating cut for SF.
Let $X\subseteq U^\pm$ be a maximum violating cut for SF.
Without loss of generality, we can assume that $X=\overline{X}$ by \eqref{eq:cut_monotone} and t-monotonicity of $\rho$.
Thus we can take $Y',Z'\subseteq U$ so that $X=Y'^+\cup Z'^-$.
From \eqref{eq:cut_reduce} we have
\begin{align}
    \hat{\kappa}_{\underline{c},\overline{c}}(Y',Z')-\beta(Y',Z')
        &=\kappa_{\underline{b},\overline{b}}(Y'^+\cup Z'^-)-\rho(Y'^+\cup Z'^-)\\
        &>\kappa_{\underline{b},\overline{b}}(Y^+\cup Z^-)-\rho(Y^+\cup Z^-)\\
        &=\hat{\kappa}_{\underline{c},\overline{c}}(Y,Z)-\beta(Y,Z).
\end{align}
Thus $(Y,Z)$ is not a maximum violating cut; a contradiction.
The converse direction can be shown similarly.
\end{proof}

A bisubmodular function $\beta:3^U\rightarrow\overline{\mathbb{R}}$ is \emph{separable} if
there exists a partition $(U_1,\dotsc,U_\ell)$ of $U$ and bisubmodular functions $\beta_t$ on each $U_t$ ($t=1,\dotsc,\ell$)
such that $\beta(Y,Z)=\sum_{t=1}^\ell \beta_t(Y\cap U_t,Z \cap U_t)$ for any $(Y,Z)\in 3^U$.

\begin{theorem}
\label{thm:bisubmod_flow}
\begin{enumerate}
\renewcommand{\labelenumi}{{\rm (\arabic{enumi}).}}
\item Suppose that $\underline{c}$ and $\overline{c}$ are integral,
and $\beta$ is reducible to an integral submodular function $\rho$.
Then we can obtain a half-integral feasible flow or a maximum violating cut
for the instance $((U,E;\partial),\underline{c},\overline{c},\beta)$ of BF in $O(\mathrm{SF}(n,m,\eta))$ time,
where $n=\lvert U\rvert$, $m=\lvert E\rvert$, and $\eta$ is the time complexity of an exchange capacity oracle for $\rho$.
\item In addition, if $\beta$ is separable to bisubmodular functions $\beta_t$ on a partition $(U_1,\dotsc,U_\ell)$
which are reducible to integral submodular functions $\rho_t$ $(t=1,\dotsc,\ell)$,
the above half-integral feasible flow $\psi$ can be chosen so that $(\nabla \psi)(U_t)\in\mathbb{Z}$ for all $t=1,\dotsc,\ell$.
\end{enumerate}
\end{theorem}

\begin{proof}
(1) As was noted after Theorem~\ref{thm:submod_feasibility},
we can obtain, in $O(\mathrm{SF}(n,m,\eta))$ time,
an integral feasible flow $\varphi$ or a maximum violating cut for the instance $((U^\pm,A),\underline{b},\overline{b},\rho)$ of SF.
In the former case, we can obtain an integral feasible flow $\psi$ for BF by \eqref{eq:give_bisubflow_by_subflow}.
In the latter case, we can obtain a maximum violating cut for BF by Proposition~\ref{prop:reduction_bisub} (2).

\vspace{5pt}
(2) We can take $\rho$ as $\rho(X)=\sum_{t=1}^\ell \rho_t(X\cap U_t^\pm)$.
Since $\rho(U_t^\pm)=\rho_t(U_t^\pm)=0$,
the above integral feasible flow $\varphi$ satisfies $(\nabla \varphi)(U_t^\pm)=0$ for each $t=1,\dotsc,\ell$.
Thus by \eqref{eq:relation_nabla}, we have
\[
    (\nabla\psi)(U_t)=\sum_{i\in U_t}(\nabla\psi)(i)=\frac{\sum_{i\in U_t} (\nabla\varphi)(i^+)-\sum_{i\in U_t}(\nabla\varphi)(i^-)}{2}
    =(\nabla\varphi)(U_t^+)
\]
for $t=1,\dotsc,\ell$. The integrality of $\varphi$ implies the lemma.
\end{proof}

\section{Minimum-cost multiflow problem}
\label{sec:mincost}

In this section, we study our problem MNMF.
We first introduce basic multiflow terminologies and the formal
description of MNMF (Section~\ref{subsec:form}).
Next we explain a combinatorial duality theory for MNMF
(Section~\ref{subsec:duality}),
and explain how the optimality check of the dual and the recovery of
an optimal multiflow
become a reducible BF (Section~\ref{subsec:ext_bidir}).
Then, in Section~\ref{subsec:prop_h}, we explain that the dual objective function
is viewed as an L-convex function on $\mathbb{G}^n$,
and present a pseudopolynomial-time algorithm for MNMF by implementing SDA.
We mention a sensitivity theorem of the dual.
Finally, in Section~\ref{subsec:scaling}, by incorporating a
cost-scaling technique
we improve the pseudopolynomial-time algorithm to a polynomial-time
one, which proves Theorem~\ref{thm:main}.

\subsection{Problem formulation}
\label{subsec:form}

Let $((V,E),S,c,d)$ be a network of an undirected graph $(V,E)$, terminal set $S \subseteq V$,
a nonnegative integer-valued capacity function $c:V\setminus S\rightarrow\mathbb{Z}_+$ on nonterminal nodes,
and a nonnegative integer-valued cost function $d:E\rightarrow\mathbb{Z}_+$ on edges.
We can assume that the graph is simple, and there are no edges connecting two terminals.
Also we can assume that $\lvert S\rvert\geq 3$.

An \emph{$S$-path} is a path connecting distinct terminals.
For a pair $f=(\mathcal{P},\lambda)$ of a set $\mathcal{P}$ of $S$-paths
and a flow-value function $\lambda:\mathcal{P}\rightarrow\mathbb{R}_+$,
the flow-values on node $i$ and edge $e$ are denoted by
$f(i):=\sum\{\lambda(P) \mid P\in\mathcal{P},i\in V(P)\}$
and $f(e):=\sum\{\lambda(P)\mid P\in\mathcal{P},e\in E(P)\}$, respectively.
A (node-capacitated) \emph{multiflow} is a pair $f=(\mathcal{P},\lambda)$
satisfying the node-capacity constraint $f(i)\leq c(i)$ for each $i\in V\setminus S$.
The \emph{total flow-value} $\mathrm{val}(f)$ of a multiflow $f=(\mathcal{P},\lambda)$
is defined by $\mathrm{val}(f):=\sum_{P\in \mathcal{P}}\lambda(P)$.

A multiflow $f=(\mathcal{P},\lambda)$ is called \emph{maximum}
if it has the maximum total flow-value among all possible multiflows.
The \emph{(edge-)cost} $d(f)$ of a multiflow $f$ is defined by $d(f):=\sum_{e\in E} d(e)f(e)$.
(We deal with the edge-cost setting, since
	node-costs are transformed to edge-costs with no complexity overhead.)
The \emph{minimum-cost node-capacitated multiflow problem (MNMF)} is
the problem of finding a minimum-cost maximum multiflow.
A multiflow $f=(\mathcal{P},\lambda)$ is said to be \emph{integral} if $\lambda$ is integer-valued,
and \emph{half-integral} if $2\lambda$ is integer-valued.

Let $n:=\lvert V\rvert$ and $m:=\lvert E\rvert$.
Suppose that $S=\{1,\dotsc,k\}$ and $V=\{1,\dotsc,n\}$.
For a technical issue, we assume that the cost $d$ are positive
by essentially the same modification as that given in \cite{Goldberg1997Scaling,Karzanov1994Minimum}; see Remark~\ref{rem:poscost}.
As Karzanov~\cite{Karzanov1979minimum} did for the edge-capacitated version (MEMF),
we consider the maximization of
\begin{equation}
\label{obj:M}
    M\cdot\mathrm{val}(f)-d(f)
\end{equation}
for a sufficiently large integer $M>0$; see Remark~\ref{rem:suff_large_M}.
Then an optimal multiflow for \eqref{obj:M} is a minimum-cost maximum multiflow for MNMF.

\begin{remark}
\label{rem:suff_large_M}
Let $C:=\max\{c(i)\mid i\in V\}$ and $D:=\max\{d(e)\mid e\in E\}$.
We show that $M> 2CD$ is sufficient in \eqref{obj:M}.
From the primal half-integrality shown in \cite{Babenko2012Min},
there exists a half-integral optimal multiflow for MNMF.
Let $f^*$ be such a half-integral optimal multiflow,
and $f$ an half-integral multiflow with the maximum value of \eqref{obj:M}.
If $\mathrm{val}(f^*)>\mathrm{val}(f)$, then the difference is at least $1/2$ and
\begin{equation}
    M(\mathrm{val}(f^*)-\mathrm{val}(f))-d(f^*)+d(f)\geq M/2-CD > 0.
\end{equation}
It contradicts the maximality of $f$.
Thus we have $\mathrm{val}(f^*)=\mathrm{val}(f)$, implying that $f$ is also optimal for MNMF.
\end{remark}

\begin{remark}
\label{rem:poscost}
The positiveness of $d$ is technically essential.
We perturb the edge-cost $d$ as follows.
Let $Z:=\{e\in E\mid d(e)=0\}$.
Define an edge-cost $d'$ by
\begin{equation}
	d'(e):=\begin{cases}1 & \text{if $e\in Z$},\\
		(2C\lvert Z\rvert+1)d(e) & \text{otherwise}.
		\end{cases}
\end{equation}
Let $f^*$ be a half-integral optimal multiflow for MNMF under the edge-cost $d'$.
From the primal half-integrality shown in \cite{Babenko2012Min},
there exists a half-integral optimal multiflow under the edge-cost $d$.
Therefore, it is sufficient to show that $f^*$ is also optimal under the edge-cost $d$.
Let $f$ be any half-integral multiflow with $\mathrm{val}(f)=\mathrm{val}(f^*)$.
By the assumption we have $d'(f^*)\leq d'(f)$. Thus,
\[
(2C\lvert Z\rvert+1)(d(f^*)-d(f))=d'(f^*)-d'(f)-f^*(Z)+f(Z)\leq C\lvert Z\rvert.
\]
Since $d(f^*)-d(f)$ is half-integral, $d(f^*)-d(f)\leq C\lvert Z\rvert/(2C\lvert Z\rvert+1)<1/2$ implies $d(f^*)\leq d(f)$.
\end{remark}

\subsection{Duality}
\label{subsec:duality}

First of all, we double the objective function \eqref{obj:M} by a technical reason.
Formally, we consider the following problem:
\begin{align}
\text{(PMF)\quad Maximize}\quad &\sum_{P\in\mathcal{P}} 2M\cdot\lambda(P)-\sum_{e\in E}2d(e)f(e)\\
\text{subject to}\quad & \text{$f=(\mathcal{P},\lambda)$ is a multiflow}.
\end{align}
The first author of this paper established~\cite{Hirai2013Half} a combinatorial duality theory for this type of multiflow problems; see also \cite{Hirai2018Dual}.
He first took a natural LP-dual, and after some reductions, gave the following dual problem:
\begin{subequations}
\begin{align}
\label{eq:dual_mnmf_a}\text{(DMF)\quad Minimize}\quad &\sum_{i\in V\setminus S} 2c_i y_i\\
\label{eq:dual_mnmf_b}\text{subject to}\quad & (x,y)=((x_1,y_1),\dotsc,(x_n,y_n))\in \mathbb{G}^n,\\
\label{eq:dual_mnmf_c}& \mathrm{dist}(x_i,x_j)-y_i-y_j\leq 2d_{ij}\quad (ij\in E),\\
\label{eq:dual_mnmf_d}& (x_s,y_s)=((M,s),0)\quad (s\in S),\\
\label{eq:dual_mnmf_e}& y_i\geq 0\quad (i\in V).
\end{align}
\end{subequations}
\noeqref{eq:dual_mnmf_a,eq:dual_mnmf_b,eq:dual_mnmf_c,eq:dual_mnmf_d,eq:dual_mnmf_e}In intuition,
a pair $(x_i,y_i)\in \mathbb{G}$ represents a ball on $\mathbb{T}$ whose center is $x_i$ and radius is $y_i$.
The condition \eqref{eq:dual_mnmf_c} means that
the distance between two balls is at most $2d_{ij}$
if the corresponding nodes $i,j$ are connected in $(V,E)$,
and \eqref{eq:dual_mnmf_d} is a boundary condition.
We say that $(x,y)\in \mathbb{G}^n$ is a \emph{potential}
if $(x,y)$ satisfies all the conditions \eqref{eq:dual_mnmf_c}--\eqref{eq:dual_mnmf_e}.

\begin{theorem}[\cite{Hirai2013Half}]
\label{thm:maxflow_minpot}
The optimum values of PMF and DMF coincide.
\end{theorem}

The following is the complementary slackness conditions.

\begin{lemma}[\cite{Hirai2013Half}]
\label{lem:slack}
Let $f=(\mathcal{P},\lambda)$ a multiflow and $(x,y)\in \mathbb{G}^n$ be a potential.
Then $f$ and $(x,y)$ are both optimal if and only if the following condition holds:
\begin{align}
    \mathrm{dist}(x_i,x_j)-y_i-y_j&=2d_{ij} && (ij\in E : f(ij)>0),\label{cond:slack1}\\
    f(i)&=c(i) && (i\in V\setminus S : y_i>0),\label{cond:slack2}\\
    \sum_{t=1}^\ell \mathrm{dist}(x_{i_{t-1}},x_{i_t})&=\mathrm{dist}(x_{i_0},x_{i_\ell})&& (P=(i_0,\dotsc,i_{\ell})\in\mathcal{P} : f(P)>0).\label{cond:slack3}
\end{align}
\end{lemma}
The condition \eqref{cond:slack3} says that
by embedding $i\mapsto x_i$, the path $P$ in $(V,E)$ becomes a shortest path in $\mathbb{T}$.

The following lemma gives a range where an optimal potential exists.

\begin{lemma}[\cite{Hirai2018Dual}]
\label{lem:range}
There exists an optimal potential $(x,y)\in \mathbb{G}^n$ with $\mathrm{dist}(0,x_i)\leq M$ and $y_i\leq 2M$ for any $i\in V\setminus S$.
\end{lemma}

\begin{proof}
Let $(x',y')=((x',s'),y')=(((x'_1,s'_1),y'_1),\dotsc,((x'_n,s'_n),y'_n))\in \mathbb{G}^n$ be an optimal potential.
We define $(x,y)\in \mathbb{G}^n$ by
\[
    (x_i,y_i):=\begin{dcases*}
    ((x'_i,s'_i),y'_i) & if $x'_i\leq M$,\\
    ((M,s'_i),y'_i) & if $x'_i>M$ and $(x'_i,y'_i)$ is integral,\\
    ((M-1/2,s'_i),y'_i) & if $x'_i>M$ and $(x'_i,y'_i)$ is non-integral
    \end{dcases*}\quad (i\in V\setminus S),
\]
and $(x_s,y_s)=((M,s),0)$ for $s\in S$.
Since $\mathrm{dist}(x_i,x_j)\leq \mathrm{dist}(x'_i,x'_j)$ for any $ij\in E$,
we see that $(x,y)$ is a new potential with the same dual objective function value as $(x',y')$.
Also, since $\mathrm{dist}(x_i,x_j)$ is at most $2M$ for any $ij\in E$,
we can modify $y_i$ so that $y_i\leq 2M$ holds for any $i\in V$.
Now the lemma follows.
\end{proof}

The above proof shows that any potential $(x,y)$ can be replaced,
without changing the objective function value,
by a potential $(x',y')$ with $\mathrm{dist}(0,x'_i)\leq M$.
Therefore we assume in the following that a potential $(x,y)$ satisfies
\begin{equation}\label{eq:range}
\mathrm{dist}(0,x_i)\leq M \quad (i \in V).
\end{equation}

\subsection{Optimality via bisubmodular flow}
\label{subsec:ext_bidir}

In this subsection we explain that the optimality of a given potential
is established by a reducible BF, which is outlined as follows.
Given a potential $(x,y)\in \mathbb{G}^n$,
we try to find a multiflow $f$ satisfying the slackness
\eqref{cond:slack1}, \eqref{cond:slack2}, and \eqref{cond:slack3};
if it exists, then both $f$ and $(x,y)$ are guaranteed to be optimal.
Observe that \eqref{cond:slack1} and \eqref{cond:slack2}
are conditions for the support $(f(e))_{e \in E}$ of a multiflow $f$.
In fact,
the third condition \eqref{cond:slack3} can also be written as a
condition for the support of $f$.
Then the existence of a multiflow $f$ satisfying the slackness with $(x,y)$
is equivalent to the existence of
a certain edge-weight $\psi$, which we call an {\em $(x,y)$-feasible support},
on a bidirected network associated with $(x,y)$.
Any $(x,y)$-feasible support $\psi$ is actually the support of some
optimal multiflow $f$.
Then $f$ can be recovered from $\psi$ in polynomial time (Lemma~\ref{lem:admi}).
Moreover, an $(x,y)$-feasible support is precisely a feasible flow
of a reducible BF (Corollary~\ref{cor:supp_algo}).

Let $(x,y)\in \mathbb{G}^n$ be a potential.
As mentioned above, we are interested in the support
of a multiflow satisfying the slackness \eqref{cond:slack1},
\eqref{cond:slack2}, and \eqref{cond:slack3} with $(x,y)$.
Motivated by \eqref{cond:slack1}, define an edge subset $E_{=} \subseteq E$ by
\begin{equation}
E_= := \{ij \in E \mid \mathrm{dist}(x_i,x_j)-y_i-y_j=2d_{ij}\},
\end{equation}
where other edges are not used by our target multiflows.
From the subnetwork $((V,E_{=}),S,c)$ and $(x,y)$,
we construct a bidirected network $((U,\tilde E;\partial),\underline{c},\overline{c})$ as follows.
See Figure~\ref{fig:bisub_rep} for the following construction.
In the figure, $x_i=0$, $x_j,x_{j'}\neq 0$, $y_j=0$, and $y_{j'}\neq 0$.

\begin{figure}[t]
\centering
\includegraphics[scale=0.8]{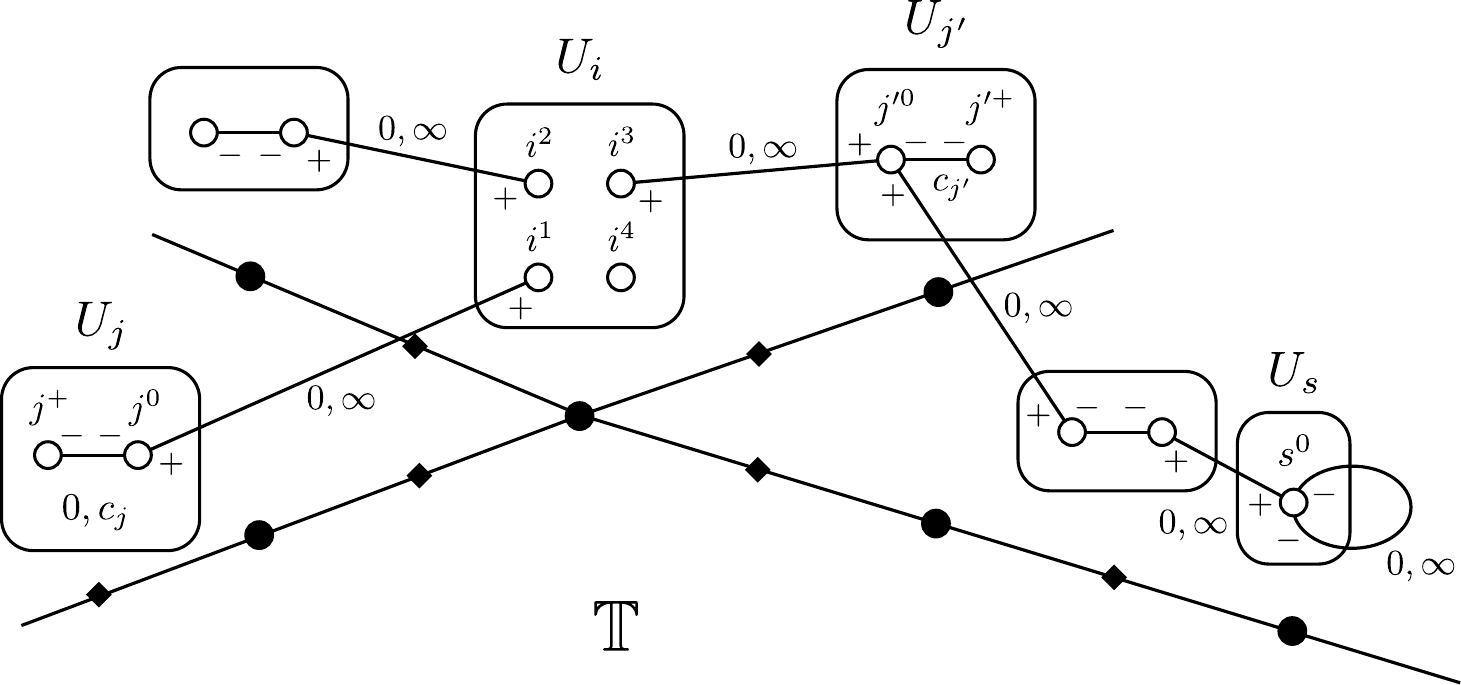}
\caption{Construction of the bidirected network.}
\label{fig:bisub_rep}
\end{figure}

For a nonterminal node $i \in V \setminus S$,
if $x_i = 0$, then consider $k$-node set $U_i := \{i^1,i^2,\ldots,i^k\}$,
where $i^s$ corresponds to the $s$-branch for $s \in S$.
If $x_i \neq 0$, then the node set $U_i := \{i^0,i^+\}$ consists of
$i^0$ (node directed to the origin) and $i^+$ (node away from the origin).
For a terminal $s \in S$, define $s^0:=s$ and $U_s := \{s^0\}$.
The node set $U$ of the bidirected network is defined as the
(disjoint) union of $U_i$ for all $i \in V$.
Next we define the edge set $\tilde E$.
Since $d_{ij}$ is positive,
we see that $x_i$ and $x_j$ are distinct points for $ij \in E_=$.
For each edge $ij\in E_=$, we replace the endpoints $i,j$
according to the position of $x_i$ and $x_j$ in $\mathbb{T}$.
If $x_i=0$ and $x_j$ is on $s$-th branch of $\mathbb{T}$,
replace $ij$ by $i^sj^0$.
If $x_i,x_j\neq 0$ are on distinct branches, replace $ij$ by $i^0j^0$.
If $x_i,x_j\neq 0$ are on the same branch
and $\mathrm{dist}(x_i,0)<\mathrm{dist}(x_j,0)$, replace $ij$ by $i^+j^0$.
Note that $j$ may be a terminal but $i$ is a nonterminal by \eqref{eq:range}.
The replaced edge set is also denoted by $E_{=}$.
Also, we add an edge $i^0i^+$ for each nonterminal node $i\in
V\setminus S$ with $x_i\neq 0$,
and add an edge $s^0 s^0$, which is a self-loop, for each terminal $s\in S$.
Let $E_-:=\{i^0i^+\mid i\in V\setminus S,\ x_i\neq 0\}$ and $E_S:=\{s^0
s^0\mid s\in S\}$.
Then the edge set $\tilde E$ of the bidirected network is defined
as $E_{=} \cup E_- \cup E_S$.
A boundary operator $\partial:\tilde{E}\rightarrow \mathbb{Z}^U$ is defined by
$\partial(ij) :=\chi_i+\chi_j$ if $ij\in E_=$ and $\partial(ij)
:=-\chi_i-\chi_j$ otherwise.
Define a lower capacity
$\underline{c}:\tilde{E}\rightarrow\underline{\mathbb{R}}$
and an upper capacity $\overline{c}:\tilde{E}\rightarrow
\overline{\mathbb{R}}$ by
\begin{equation}
    \underline{c}(e):=\begin{dcases*}
        c(i) & if $e\in E_-$ and $y_i>0$,\\
        0 & otherwise,
    \end{dcases*}\quad
    \overline{c}(e):=\begin{dcases*}
        c(i) & if $e\in E_-$,\\
        \infty & otherwise
    \end{dcases*}
\end{equation}
for each edge $e\in \tilde{E}$.

We say that a half-integral edge-weight function
$\psi:\tilde{E}\rightarrow \mathbb{Z}/2$
is an \emph{$(x,y)$-feasible support}
if it satisfies
\begin{equation}
\label{cond:supp1}\underline{c}(e) \leq \psi(e) \leq \overline{c}(e)
\end{equation}
for each edge $e\in \tilde{E}$ and
\begin{align}
\label{cond:supp2} (\nabla\psi)(i')&=0 && (x_i\neq 0,\ i'\in U_i),\\
\label{cond:supp3} (\nabla\psi)(i^\ell)-
(\nabla\psi)(U_i\setminus\{i^\ell\}) &\leq 0 && (x_i=0,\
\ell=1,\dotsc,k),\\
\label{cond:supp4}(\nabla\psi)(U_i)&\leq 2c(i) && (x_i=0,\ y_i=0),\\
\label{cond:supp5}(\nabla\psi)(U_i)&= 2c(i) && (x_i=0,\ y_i>0),\\
\label{cond:supp6}(\nabla\psi)(U_i)&\in \mathbb{Z}_+ && (x_i=0)
\end{align}
for each node $i\in V$.
See Section~\ref{subsubsec:bisub} for the definition of the boundary operator $\nabla$.

Any half-integral optimal multiflow $f$
gives rise to an $(x,y)$-feasible support from
its support $(f(e))_{e\in E}$.
Indeed, let $\psi_f:\tilde{E}\rightarrow \mathbb{Z}/2$ be a function defined by
\[
    \psi_f(i'j'):=\begin{dcases*}
        f(ij) & if $i'j'\in E_=$, $i'\in U_{i}$, and $j'\in U_{j}$,\\
        f(i) &  if $i'j'\in E_-$ and $i',j'\in U_i$,\\
        f(s)/2 & if $i'j'\in E_S$ and $i'j'=s^0s^0$
    \end{dcases*}\quad (i'j'\in \tilde{E}).
\]
Then $\psi_f$ is an $(x,y)$-feasible support. Indeed,
it is obvious that \eqref{cond:supp4} is the capacity constraint and
\eqref{cond:supp5} corresponds to \eqref{cond:slack2}.
\eqref{cond:supp1} also corresponds to the capacity constraint and \eqref{cond:slack2}.
\eqref{cond:supp3} comes from \eqref{cond:slack3}:
In the embedding $i \mapsto x_i \in \mathbb{T}$, a flow coming
from the $\ell$-th branch must go to the other branches.
Hence the total flow-value $(\nabla \psi) (i_{\ell})$ for such flows
is at most the total flow-value $(\nabla \psi) (U_i \setminus \{i_{\ell}\})$
of flows going to branches other than the $\ell$-th branch.
The condition \eqref{cond:supp2} is also explained by the same way.
The condition \eqref{cond:supp6} follows from the half-integrality of $f$.

Conversely, we can decompose an $(x,y)$-feasible support into a
half-integral multiflow satisfying
\eqref{cond:slack1}--\eqref{cond:slack3}:

\begin{lemma}[{\cite{Hirai2013Half}}]
\label{lem:admi}
Let $(x,y)$ be a potential.
\begin{enumerate}
\renewcommand{\labelenumi}{{\rm (\arabic{enumi}).}}
\item $(x,y)$ is optimal
    if and only if an $(x,y)$-feasible support exists.
\item Given an $(x,y)$-feasible support,
    we can obtain a half-integral optimal multiflow in $O(nm)$ time.
\end{enumerate}
\end{lemma}

The algorithm for Lemma \ref{lem:admi} given in \cite[Algorithm 1]{Hirai2018Dual}
is an analogue of the flow decomposition procedure in the ordinary network flow.

Now our task is to find an $(x,y)$-feasible support.
This can be done by solving a reducible BF.
Recall that $S=\{1,\dotsc,k\}$.
Motivated by \eqref{cond:supp3} and \eqref{cond:supp4}, we consider the following condition for $z\in \mathbb{R}_+^S$ and $c\in \mathbb{R}_+$:
\begin{gather}
\label{cond:nf1}z(S)\leq 2c,\\
\label{cond:nf2}z(t)\leq z(S\setminus\{t\}) \quad (t\in S).
\end{gather}
Let $P_c$ be sets of vectors of $\mathbb{R}_+^S$ satisfying \eqref{cond:nf1} and \eqref{cond:nf2}.
We call $P_c$ as the \emph{node-flowing polytope} with capacity $c$.
Also, motivated by \eqref{cond:supp5}, we consider the condition
\begin{equation}
\label{cond:nf1t}z(S)=2c.
\end{equation}
Let $\overline{P}_c$ be sets of vectors of $\mathbb{R}_+^S$ satisfying \eqref{cond:nf1t} and \eqref{cond:nf2}.
We call $\overline{P}_c$ as the \emph{tight node-flowing polytope} with capacity $c$.


\begin{lemma}
\label{lem:node_flowing}
$P_c$ and $\overline{P}_c$ are bisubmodular polyhedra
of reducible bisubmodular functions $\beta_c$ and $\overline{\beta}_c$, respectively,
where $\beta_c, \overline{\beta}_c:3^S\rightarrow \mathbb{R}$ are reducible to submodular functions whose exchange capacity oracles are calculated in $O(k)$ time.
\end{lemma}

The proof is given in Section~\ref{subsec:node_flowing}.

\begin{corollary}
\label{cor:supp_algo}
Let $(x,y)$ be an optimal potential.
Then we can find an $(x,y)$-feasible support in $O(\mathrm{SF}(kn,m,k))$ time.
\end{corollary}

\begin{proof}
On the network $((U,\tilde{E};\partial),\underline{c},\overline{c})$, define a reducible bisubmodular function $\beta:3^U\rightarrow\mathbb{R}$ by
\begin{equation}
\label{def:beta}
\begin{aligned}
\beta(Y,Z)&:=\sum_{i\in V:x_i=0,y_i=0} \beta_{c_i}(Y\cap U_i,Z\cap U_i)\\
	&\qquad+\sum_{i\in V:x_i=0,y_i>0} \overline{\beta}_{c_i}(Y\cap U_i,Z\cap U_i)\quad ((Y,Z)\in 3^U).
\end{aligned}
\end{equation}
There exists an $(x,y)$-feasible support by the optimality of $(x,y)$ and Lemma~\ref{lem:admi}.
Since an $(x,y)$-feasible support is a feasible flow of
the instance $((U,\tilde{E};\partial),\underline{c},\overline{c},\beta)$ of BF,
this instance is feasible.

By Theorem~\ref{thm:bisubmod_flow} and Lemma~\ref{lem:node_flowing},
we can obtain, in $O(\mathrm{SF}(kn,m,k))$ time, a half-integral feasible flow $\psi:\tilde{E}\rightarrow\mathbb{R}$
such that $(\nabla \psi)(U_i)\in\mathbb{Z}$ for $i\in V$ with $x_i=0$.
Moreover, by the definition of $\underline{c}$ and $\partial$, this can be strengthened so that $(\nabla \psi)(U_i)\in\mathbb{Z}_+$.
Observe that any feasible flow satisfies \eqref{cond:supp1}--\eqref{cond:supp5}.
Also, the obtained feasible flow satisfies the half-integrality and \eqref{cond:supp6}.
Hence, this is an $(x,y)$-feasible support.
\end{proof}

Therefore, given an optimal potential,
we can obtain a half-integral optimal multiflow of the original problem in $O(\mathrm{SF}(kn,m,k)+nm)=O(\mathrm{SF}(kn,m,k))$ time.



\subsection{Discrete convexity and sensitivity}
\label{subsec:prop_h}

Here we study the objective function of the dual DMF.
Define a function $h=h_{d,M}:\mathbb{G}^n\rightarrow \overline{\mathbb{R}}$ by
\begin{equation}
\label{eq:h}
h(x,y):=\begin{cases} \sum_{i\in V\setminus S} c_i y_i & \text{if $(x,y)$ is a potential for $d$ and $M$,} \\
\infty & \text{otherwise}\end{cases}\quad ((x,y)\in \mathbb{G}^n).
\end{equation}
Then DMF is precisely the minimization of $h$.

\begin{proposition}[{\cite[Proposition 3.6]{Hirai2018Dual}}]
\label{prop:L_conv}
$h$ is L-convex on $\mathbb{G}^n$.
\end{proposition}

Thus we can apply SDA to $h$ to obtain an optimal potential.
The initial step is finding a potential $(x,y)\in\mathbb{G}^n$.
It can be done immediately since
\begin{equation}
\label{eq:init}
    (x_i,y_i)=\begin{cases}
    ((M,s),0) & \text{if $i=s\in S$,} \\
    (0,M) & \text{otherwise} \\
    \end{cases}\quad (i\in V)
\end{equation}
is a potential for any $d$ and $M$.
To find a steepest descent direction at a current point $(x,y)\in \mathbb{G}^n$,
we again make use of the bisubmodular flow theory.


Suppose that $(x,y)$ is not optimal.
By Lemma~\ref{lem:admi}, there exists no $(x,y)$-feasible support on the bidirected network $((U,\tilde{E};\partial),\underline{c},\overline{c})$.
Thus BF introduced in the previous subsection is infeasible.
Then by Theorem~\ref{thm:bisubmod_feasibility}, there exists a violating cut in the network.
From a maximum violating cut, we can find a steepest descent direction.

\begin{lemma}
\label{lem:sdd_algo}
Let $(x,y)$ be a potential.
Suppose that $(x,y)$ is not optimal.
Then we can find a steepest descent direction of $h$ at $(x,y)$ in $O(\mathrm{SF}(kn,m,k))$ time.
\end{lemma}

We give the proof in Section~\ref{subsec:sdd}.
By simply combining Corollary~\ref{cor:supp_algo} and Lemma~\ref{lem:sdd_algo}, we obtain a pseudopolynomial-time algorithm for MNMF.

\begin{proposition}
\label{prop:pseudopoly}
A half-integral multiflow which maximizes \eqref{obj:M} can be obtained in $O(M\cdot \mathrm{SF}(kn,m,k))$ time.
\end{proposition}

\begin{proof}
By Lemma \ref{lem:range}, there exists an optimal potential $p^*=(x^*,y^*)\in \mathbb{G}^n$ for DMF
satisfying $\lVert p-p^*\rVert=O(M)$, where $p=(x,y)$.
Thus by Theorem \ref{thm:SDA}, SDA terminates after at most $O(M)$ iterations.
Each iteration can be done in $O(\mathrm{SF}(kn,m,k))$ time by Lemma~\ref{lem:sdd_algo}.
Also, once an optimal potential is obtained, we can find a half-integral optimal multiflow
in $O(\mathrm{SF}(kn,m,k))$ time by Lemma~\ref{lem:admi} and Corollary~\ref{cor:supp_algo}.
\end{proof}

From Remarks~\ref{rem:suff_large_M} and \ref{rem:poscost},
we see that our algorithm in Proposition~\ref{prop:pseudopoly} runs in $O(nC^2D\cdot\mathrm{SF}(kn,m,k))$ time.

To combine the above algorithm with cost-scaling,
we use the following sensitivity property of $h$ under a small cost change.

\begin{theorem}[Sensitivity]
\label{thm:sensitivity}
For a positive edge-cost $d:E\rightarrow \mathbb{Z}_+$, let $e_0\in E$ be an edge with $d(e_0)\in 2\mathbb{Z}$.
Define an edge-cost $d_1:E\rightarrow\mathbb{Z}_+$ by
\[
    d_1(e):=\begin{dcases*} d(e_0)-1 & if $e=e_0$,\\
        d(e) & otherwise
    \end{dcases*}\quad (e\in E).
\]
Let $p\in\mathbb{G}^n$ be a minimizer of $h_{d,M}$.
Then there exists a minimizer $q\in\mathbb{G}^n$ of $h_{d_1,M}$ satisfying $\lVert p-q\rVert\leq 2$.
\end{theorem}

The proof is given in Section~\ref{subsec:sensitivity}.

\subsection{Cost scaling algorithm}
\label{subsec:scaling}

The pseudopolynomial-time algorithm we gave in the previous subsection consists of two stages.
In the first stage, one tries to minimize the L-convex function $h_{d,M}$ on $\mathbb{G}^n$ by SDA.
In the second stage, construct the desired optimal multiflow.
Both of them are done by solving a reducible BF.
The second stage can be done in strongly polynomial time (Lemma~\ref{lem:admi} and Corollary~\ref{cor:supp_algo}).
The bottleneck is in the first stage.
As we mentioned above, 
the number of iterations of SDA becomes $O(M)$ in the worst case.

We overcome this issue by the \emph{scaling} method.
Without loss of generality, we can assume that $M$ is the power of two, i.e., $M=2^\mu\ (\mu\geq 0)$.
The scaling algorithm consists of $(\mu+1)$ scaling phases.
In $t$-th phase ($t=0,1,\dotsc,\mu$), we deal with the problem of
finding a multiflow $f$ that maximizes
\[
    \frac{M}{2^{t}}\cdot \mathrm{val}(f)-\left(\left\lceil \frac{d}{2^{t}}\right\rceil\right) (f),
\]
where $(\lceil d/2^t \rceil)(e) := \lceil d(e)/2^t \rceil$ for $e\in E$.
Since $d$ is positive, $\lceil d/2^t \rceil$ is also positive.
Hence we can apply the same discussion as above to this scaled problem,
i.e., the dual objective function, $h_{\lceil d/2^{t}\rceil,M/2^{t}}$,
is an L-convex function on $\mathbb{G}^n$ and can be minimized via SDA.

In the first phase $t=\mu$, we can find an optimal potential in constant time.

\begin{lemma}
\label{lem:scale_start}
Let $(x,y)\in \mathbb{G}^n$ a potential for $\mu$-th phase defined by
\[
    (x_i,y_i)=\begin{dcases*}
        (0,0) & if $i\in V\setminus S$,\\
        ((1,s),0) & if $i=s\in S$
    \end{dcases*}\quad (i\in V).
\]
Then $(x,y)$ is an optimal potential.
\end{lemma}

\begin{proof}
It is easy to see that $(x,y)$ is a potential since the left hand side of \eqref{eq:dual_mnmf_c} is at most one and the right hand side is at least two.
The optimality follows from that its objective function value is zero.
\end{proof}

We show that given an optimal potential for $(t+1)$-st phase,
we can find an optimal potential in $t$-th phase efficiently.

\begin{lemma}
\label{lem:proximity}
Given an optimal potential for $(t+1)$-st phase $(t=0,\dotsc,\mu-1)$,
we can obtain an optimal potential in $t$-th phase
by SDA with at most $2m+4$ iterations.
\end{lemma}

\begin{proof}
It suffices to show the case $t=0$.
Let $p\in\mathbb{G}^n$ be an optimal potential in the phase $t=1$, i.e.,
$p$ is a minimizer of $h_{d'/2,M/2}$, where $d':=2\cdot\lceil d/2\rceil$.
Define $2p\in \mathbb{G}^n$ by $(2p)_i:=2p_i$ for each $i\in V$, where $2(x,y):=(2x,2y)\in\mathbb{G}$ for $(x,y) \in \mathbb{G}$. 

\begin{claim}
$2p\in \mathbb{G}^n$ is a minimizer of $h_{d',M}$.
\end{claim}

\begin{proof}
We observe from \eqref{eq:dual_mnmf_c}--\eqref{eq:dual_mnmf_e} that $2p$ is a potential under $d'$ and $M$.
Thus we have $h_{d',M}(2p)=2h_{d'/2,M/2}(p)$.
Since $p$ is a minimizer of $h_{d'/2,M/2}$,
there exists a multiflow $f$ satisfying $h_{d'/2,M/2}(p)=M/2\cdot \mathrm{val}(f)-d'(f)/2$ by the duality (Theorem \ref{thm:maxflow_minpot}).
Hence we have $h_{d',M}(2p)=2h_{d'/2,M/2}(p)=M\cdot\mathrm{val}(f)-d'(f)$,
and again by the duality, $2p$ is a minimizer of $h_{d',M}$.
\end{proof}

Observe that
\begin{equation}
    d(e)=\begin{dcases*}
        d'(e)-1 & if $d(e)$ is odd,\\
        d'(e) & if $d(e)$ is even
    \end{dcases*}
\end{equation}
for each $e\in E$.
Thus by Theorem~\ref{thm:sensitivity} (and the triangle inequality),
there exists a minimizer $q$ of $h_{d,M}$ satisfying $\lVert 2p-q\rVert\leq 2m$.

We use SDA to obtain such a potential $q$, but here $2p$ is not always a potential under the edge-cost $d$.
Recall that we can run the iteration of SDA only on a potential.
Thus in such a case, we define a new potential $r\in \mathbb{G}^n$ from $2p=(x',y')$ by
\[
    r_i:=\begin{dcases*} (x'_i,y'_i+2) & if $i\in V\setminus S$,\\
        2p_s & if $i=s\in S$
    \end{dcases*}\quad (i\in V).
\]
Then $r$ is a potential under the edge-cost $d$ (and $M$), and $\lVert 2p-r\rVert\leq 2$.
Therefore, Theorem \ref{thm:SDA} (and the triangle inequality) implies that
by starting SDA from the potential $r$,
a minimizer of $h_{d,M}$ can be obtained with at most $2m+4$ iterations.
\end{proof}

Now we obtain a weakly polynomial-time algorithm for MNMF.

\begin{proof}[Proof of Theorem \ref{thm:main}]
By Remarks~\ref{rem:suff_large_M} and \ref{rem:poscost}, we can take $\mu=O(\log(nCD))$.
As noted in Lemma~\ref{lem:scale_start}, an optimal potential for the first phase $t=\mu$ can be obtained in constant time.
Lemma~\ref{lem:proximity} implies that
we can obtain an optimal potential for DMF by $O(m\mu)$ iterations of SDA.
Each iteration can be done in $O(\mathrm{SF}(kn,m,k))$ time by Lemma~\ref{lem:sdd_algo}.
Also after SDA, an optimal half-integral multiflow for \eqref{obj:M} can be obtained in $O(\mathrm{SF}(kn,m,k))$ time
by Lemma~\ref{lem:admi} and Corollary~\ref{cor:supp_algo}.
Therefore, the overall algorithm runs in $O(m\log(nCD)\mathrm{SF}(kn,m,k))$ time.
\end{proof}

\section{Proofs}
\label{sec:proofs}

\subsection{Lift of the node-flowing polytope (Proof of Lemma~\ref{lem:node_flowing})}
\label{subsec:node_flowing}

Let $B_c$ be the set of vectors $x\in\mathbb{R}^{S^\pm}$ satisfying
\begin{align}
\label{eq:nodeflowing_ineq1}x(S^\pm)&=0,\\
\label{eq:nodeflowing_ineq2}x(S^+)&\leq 2c,\\
\label{eq:nodeflowing_ineq3}x(i^+)+x(S^-\setminus\{i^-\})&\leq 0 \quad (i\in S),\qquad\\
\label{eq:nodeflowing_ineq4}0\leq x(i^+)&\leq c \quad (i\in S),\\
\label{eq:nodeflowing_ineq5}-c\leq x(i^-)&\leq 0 \quad (i\in S).
\end{align}
Also, let $\overline{B}_c$ be the set of vectors $x\in\mathbb{R}^{S^\pm}$ satisfying \eqref{eq:nodeflowing_ineq1}--\eqref{eq:nodeflowing_ineq5} and
\begin{equation}
\label{eq:nodeflowing_ineq6}x(S^-)\leq -2c.
\end{equation}
We show that $B_c$ and $\overline{B}_c$ are the lifts of the node-flowing polytope $P_c$
and the tight node-flowing polytope $\overline{P}_c$, respectively.

\begin{lemma}
$\Phi(B_c)=P_c$, $\Psi(P_c)\subseteq B_c$,
$\Phi(\overline{B}_c)=\overline{P}_c$, and $\Psi(\overline{P}_c)\subseteq \overline{B}_c$.
\end{lemma}

\begin{proof}
For $z\in P_c$, we have $z(i)\leq c$ by $z(i)+z(i)\leq z(i)+z(S\setminus\{i\})\leq 2c$.
Thus $\Psi(z)$ satisfies \eqref{eq:nodeflowing_ineq4} and \eqref{eq:nodeflowing_ineq5}.
We can check \eqref{eq:nodeflowing_ineq1}--\eqref{eq:nodeflowing_ineq3} directly from the definition.
Thus $\Psi(z)\in B_c$.
This immediately implies $P_c\subseteq \Phi(B_c)$.
We show the converse direction.
Let $x\in B_c$ and $z:=\Phi(x)$. Then $z$ is nonnegative by \eqref{eq:nodeflowing_ineq4}\eqref{eq:nodeflowing_ineq5},
and satisfies \eqref{cond:nf1}\eqref{cond:nf2} as
\begin{equation}
\label{eq:induce_nf1}
    z(S)= \frac{x(S^+)-x(S^-)}{2}=\frac{2x(S^+)-x(S^\pm)}{2}\leq 2c,
\end{equation}
and
\begin{align}
    z(i)-z(S\setminus\{i\})
        &= \frac{x(i^+)-x(i^-)-(x(S^+\setminus\{i^+\})-x(S^-\setminus\{i^-\}))}{2}\\
        &=\frac{x(i^+)+x(S^-\setminus\{i^-\})-(x(S^+\setminus\{i^+\})+x(i^-))}{2}\\
        &=\frac{2(x(i^+)+x(S^-\setminus\{i^-\}))-x(S^\pm)}{2}\leq 0.
\end{align}

For $\overline{P}_c$ and $\overline{B}_c$, the proof is almost the same.
In fact, by noticing that
the conditions \eqref{eq:nodeflowing_ineq1}\eqref{eq:nodeflowing_ineq2}\eqref{eq:nodeflowing_ineq6} imply $x(S^+)=2c$ (and $x(S^-)=-2c$),
\eqref{eq:induce_nf1} holds with equality.
\end{proof}

Next, we show that $B_c$ and $\overline{B}_c$ are base polyhedra of t-monotone submodular functions.
This proves Lemma~\ref{lem:node_flowing}.

For $X\subseteq S^\pm$, we denote $X^+:=X\cap S^+$ and $X^-:=X\cap S^-$.
For $u,v\in\{0,1,2\}$, we say $X\subseteq S^\pm$ is of type $uv$ if
\[
    \lvert X^+\rvert\,\begin{dcases*}
        =u & if $u=0,1$,\\
        \geq u & if $u=2$,
    \end{dcases*}\qquad
    \lvert X^-\rvert\,\begin{dcases*}
        =k-v & if $v=0,1$,\\
        \leq k-v & if $v=2$.
    \end{dcases*}
\]
Additionally, $X$ is of type $11^*$ if $X^+=\{i^+\}$ and $X^-=S^-\setminus\{i^-\}$ for some $i\in S$.
In the below, when we say $X$ is of type 11, we assume that $X$ is not of type $11^*$.

Define $\rho_c,\overline{\rho}_c:2^{S^\pm}\rightarrow \mathbb{R}$ by
\begin{align}
\label{def:rho_c}
\rho_c(X)&:=\begin{dcases*}
 2c & if $X$ is of type 22,\\
 c  & if $X$ is of type 11, 12, or 21,\\
 0  & otherwise
 \end{dcases*}\quad (X\subseteq S^\pm),\\
\label{def:ex_rho_c}
\overline{\rho}_c(X)&:=c\cdot(\min\{\lvert X^+\rvert,2\}+\min\{k-2-\lvert X^-\rvert,0\})\quad (X\subseteq S^\pm).
\end{align}
We summarize the definition of $\rho_c$ and $\overline{\rho}_c$ in Table~\ref{tab:nf_submod_val}.

\begin{table}[t]
\caption{Submodular functions $\rho_c$ (left) and $\overline{\rho}_c$ (right).}
\label{tab:nf_submod_val}
\vspace{10pt}
\begin{subtable}[b]{.49\linewidth}
\centering
\begin{tabular}{c|c|c|c|c|}
\multicolumn{2}{c}{} & \multicolumn{1}{c}{\hphantom{$-2c$}} & \multicolumn{1}{c}{$\lvert X^-\rvert$} & \multicolumn{1}{c}{}
     \vrule width 0pt height 7pt depth 7pt \\ \cline{2-5}
 & & $k$ & $k-1$ & $\leq k-2$ \\ \cline{2-5}
 & $0$ & 0 & 0 & 0 \\ \cline{2-5}
$\lvert X^+\rvert$ & $1$ & 0 & 0 or $c$ & $c$ \\ \cline{2-5}
 & $\geq 2$ & 0 & $c$ & $2c$\\\cline{2-5}
\end{tabular}
\end{subtable}
\begin{subtable}[b]{.49\linewidth}
\centering
\begin{tabular}{c|c|c|c|c|}
\multicolumn{3}{c}{} & \multicolumn{1}{c}{$\lvert X^-\rvert$} & \multicolumn{1}{c}{}
     \vrule width 0pt height 7pt depth 7pt \\ \cline{2-5}
 & & $k$ & $k-1$ & $\leq k-2$ \\ \cline{2-5}
 & $0$ & $-2c$ & $-c$ & 0 \\ \cline{2-5}
$\lvert X^+\rvert$ & $1$ & $-c$ & 0 & $c$ \\ \cline{2-5}
 & $\geq 2$ & 0 & $c$ & $2c$\\\cline{2-5}
\end{tabular}
\end{subtable}
\end{table}

\begin{lemma}
\label{lem:nodeflowing_submod}
Both $\rho_c$ and $\overline{\rho}_c$ are t-monotone submodular functions,
where $B_c=\mathcal{B}(\rho_c)$ and $\overline{B}_c=\mathcal{B}(\overline{\rho}_c)$.
\end{lemma}

\begin{proof}
(submodularity)
We first consider $\overline{\rho}_c$.
Since the minimum of any monotone submodular function and a nonnegative constant is submodular,
the functions $\min\{\lvert X^+\rvert,2\}$ and $\min\{k-2-\lvert X^-\rvert,0\}$ are submodular.
Thus $\overline{\rho}_c$ is also submodular.

Next we consider $\rho_c$.
$\rho_c(\emptyset)=0$ is clear.
We denote by $\rho_c(u|X)$ the difference $\rho_c(X\cup\{u\})-\rho_c(X)$ for $X\subseteq S^\pm$ and $u\in S^\pm$.
It is well-known that $\rho_c$ is submodular if and only if it satisfies
\begin{equation}
\label{eq:dim_utility}
\rho_c(u|X)\leq \rho_c(u|X')
\end{equation}
for all $X'\subseteq X\subseteq S^\pm$ and $u\notin X$.
From Table~\ref{tab:nf_submod_val},
we can determine $\rho_c(i^+|X)\in\{0,c\}$ and $\rho_c(i^-|X)\in\{-c,0\}$ for any $i\in S$.

Suppose that $X$ is of type 00, 10, or 20.
Then $u$ must be $i^+$ for some $i\in S$, and the left hand side of \eqref{eq:dim_utility} is 0.
Thus \eqref{eq:dim_utility} holds.

Suppose that $X$ is of type 11 or 21.
Then $\rho_c(i^+|X)=0\leq \rho_c(i^+|X')\in\{0,c\}$ and $\rho_c(j^-|X)=-c\leq \rho_c(j^-|X')\in\{-c,0\}$ for $i^+,j^-\notin X$.

Suppose that $X$ is of type 01 or $11^*$.
If $X$ is of type 01, then $X'$ must be of type 02.
If $X$ is of type $11^*$, then $X'$ must be of type 01, 02, or 12.
In all cases, $\rho_c(i^+|X')=c\geq \rho_c(i^+|X)$ and $\rho_c(j^-|X')=0=\rho_c(j^-|X)$ for $i^+,j^-\notin X$.

Suppose that $X$ is of type 02, 12, or 22.
Then $X'$ must also be of type 02, 12, or 22.
Let $i^+\notin X$.
We may consider the case $\rho_c(i^+|X')=0$.
Then $X'$ and $X$ must be of type 22, and thus $\rho_c(i^+|X)=0$.
Let $j^-\notin X$.
We may consider the case $\rho_c(j^-|X')=-c$.
Then $X'$ is of type 12 or 22, and $\lvert X'^-\rvert =k-2$.
In both cases, $X$ is of type 22 and $\lvert X^-\rvert=k-2$, and thus $\rho_c(j^-|X)=-c$.

\vspace{5pt}\noindent
(t-monotonicity) For any $X\subseteq S^\pm$ with $X\neq\underline{X}$, the type of $\underline{X}$ is 01, 02, 12, 22.
If $\underline{X}$ is of type 01, then $X$ is of type 10.
If $\underline{X}$ is of type 02, then $X$ is of type 11, 12, 20, 21, or 22.
If $\underline{X}$ is of type 12, then $X$ is of type 21 or 22.
If $\underline{X}$ is of type 22, then $X$ is also of type 22.
In all cases, the t-monotonicity inequality holds for both $\rho_c$ and $\overline{\rho}_c$.

\vspace{5pt}\noindent
(base polyhedra)
Observe that the inequalities \eqref{eq:nodeflowing_ineq1}--\eqref{eq:nodeflowing_ineq5} of $B_c$ appear in
inequalities \eqref{eq:base} of $\mathcal{B}(\rho_c)$.
Hence $\mathcal{B}(\rho_c) \subseteq B_c$.
To show the converse, it suffices to
check that inequalities of $\mathcal{B}(\rho_c)$ are deduced from \eqref{eq:nodeflowing_ineq1}--\eqref{eq:nodeflowing_ineq5}.
This is a routine verification.
For example, if $X$ is of type $11^*$, then $x(X)\leq 0$ is deduced from \eqref{eq:nodeflowing_ineq3}.
If $X$ is of type 22, then $x(X) \leq 2c$ is deduced
from substituting \eqref{eq:nodeflowing_ineq2}, \eqref{eq:nodeflowing_ineq4}, \eqref{eq:nodeflowing_ineq5}
for $x(X) = x(S^+) - x(S^+ \setminus X^+) + x(X^-)$.
$\overline{B}_c = \mathcal{B}(\overline{\rho}_c)$ is similarly shown.
\end{proof}

By Lemma~\ref{lem:reducible},
the bisubmodular functions $\beta_c,\overline{\beta}_c:3^{S}\rightarrow\mathbb{R}$ in Lemma~\ref{lem:node_flowing}
are as follows:
\begin{align}
\label{def:beta_c}
\beta_c(Y,Z) &:= \begin{dcases*}
	2c & if $\lvert Y \rvert \geq 2$,\\
	c & if $\lvert Y \rvert = 1$ and $\lvert Z\rvert \leq k-2$,\\
	0 & otherwise ($Y = \emptyset$ or $\lvert Z \rvert \geq k-1$)
\end{dcases*}\quad (Y,Z)\in 3^{S},\\
\label{def:ex_beta_c}
\overline{\beta}_c(Y,Z) &:= \begin{dcases*}
	2c & if $\lvert Y \rvert \geq 2$,\\
	c & if $\lvert Y \rvert = 1$ and $\lvert Z\rvert \leq k-2$,\\
    -c & if $Y=\emptyset$ and $\lvert Z\rvert=k-1$,\\
    -2c & if $Y=\emptyset$ and $\lvert Z\rvert = k$,\\
    0 & otherwise ($\lvert Y\rvert=1$ and $\lvert Z\rvert= k-1$,\\
        &\qquad\qquad or $Y=\emptyset$ and $\lvert Z\rvert\leq k-2$)
\end{dcases*}\quad (Y,Z)\in 3^{S}.
\end{align}

\subsection{Proof of sensitivity theorem (Theorem~\ref{thm:sensitivity})}
\label{subsec:sensitivity}

For the proof, we introduce a function $\pi:(\mathbb{T}\times\mathbb{R})^2\rightarrow\mathbb{R}$ defined by
\[
    \pi((x,y),(x',y')):=\mathrm{dist}(x,x')-y-y'\quad ((x,y),(x',y')\in \mathbb{T}\times\mathbb{R}).
\]
We say that $(x,y)\in \mathbb{G}^n$ is \emph{nonnegative} if $y_i\geq 0$ for all $i\in V$.
Then a nonnegative point $(x,y)\in \mathbb{G}^n$ is a potential under a positive integral edge-cost $d$ and a positive integer $M$
if and only if it satisfies \eqref{eq:dual_mnmf_d} and
\begin{equation}
\label{cond:pot_h} \pi((x_i,y_i),(x_j,y_j))\leq 2d(ij)\quad (ij\in E).
\end{equation}
Also, for a positive half-integral edge-cost $d$ and a positive integer $M$,
we say a nonnegative point $(x,y)\in \mathbb{G}^n$ is a \emph{potential} if it satisfies \eqref{eq:dual_mnmf_d} and \eqref{cond:pot_h}.

For $u,v\in \mathbb{G}$, let $u\sqcap v:=\lceil (u+v)/2\rceil$ and $u\sqcup v:=\lfloor (u+v)/2\rfloor$.
The operators $\sqcap$, $\sqcup$ are extended naturally to $\mathbb{G}^n$.
We use the following lemmas.

\begin{lemma}[{\cite[in Lemmas 3.7 and 3.8]{Hirai2018Dual}}]
\label{lem:property_h}
Let $u,v,u',v'\in \mathbb{G}$.
\begin{enumerate}
\renewcommand{\labelenumi}{{\rm (\arabic{enumi}).}}
\item\label{p_h_int} $\pi(u,v)$ is an integer.
\item\label{p_h_conv} $\pi((u+u')/2,(v+v')/2)$ is a half-integer and satisfies $\pi(u,v)+\pi(u',v')\geq 2\pi((u+u')/2,(v+v')/2)$.
\item\label{p_h_near} Let $\Delta:=\pi(u\sqcap u',v\sqcap v')-\pi((u+u')/2,(v+v')/2)$
    and $\Delta':=\pi(u\sqcup u',v\sqcup v')-\pi((u+u')/2,(v+v')/2)$.
    Then $\Delta,\Delta'\in\{-1,-1/2,0,1/2,1\}$.
\item\label{p_h_sym} If $\pi((u+u')/2,(v+v')/2)\geq 1/2$, then $\Delta=-\Delta'$.
\end{enumerate}
\end{lemma}

\begin{lemma}
\label{lem:d_property}
Let $p,q\in\mathbb{G}^n$ with $\lVert p-q\rVert\geq 2$.
Then $p\sqcap q\notin \{p,q\}$ and $p\sqcup q\notin\{p,q\}$.
Moreover, $\lVert p-p\sqcap q\rVert<\lVert p-q\rVert$ and $\lVert p-p\sqcup q\rVert<\lVert p-q\rVert$.
\end{lemma}

\begin{proof}
Let $i$ be any index with $\lVert p_i-q_i\rVert\geq 2$.
By the definition we have $\lVert p_i - (p_i+q_i)/2\rVert = \lVert p_i-q_i\rVert / 2$.
Also we have $\lVert (p_i+q_i)/2-(p_i\sqcap q_i)\rVert \leq 1/2$.
Then
\begin{align*}
\lVert p_i - p_i\sqcap q_i\rVert &\leq \lVert p_i - (p_i+q_i)/2\rVert + \lVert (p_i+q_i)/2 - (p_i\sqcap q_i)\rVert\\
    &\leq \lVert p_i-q_i\rVert/2 +1/2 < \lVert p_i-q_i\rVert.
\end{align*}
Similarly, it holds that $\lVert p_i - p_i\sqcup q_i\rVert < \lVert p_i-q_i\rVert$. These imply the lemma.
\end{proof}

\begin{proof}[Proof of Theorem~\ref{thm:sensitivity}]
Define an edge-cost $d_{1/2}:E\rightarrow \mathbb{Z}/2$ by
\begin{equation}
    d_{1/2}(e):=\begin{dcases*}
        d(e_0)-1/2 & if $e=e_0$,\\
        d(e) & otherwise.
    \end{dcases*}
\end{equation}
Our proof consists of two parts: we show that
(i) there exists a minimizer $q\in \mathbb{G}^n$ of $h_{d_{1/2},M}$ satisfying $\lVert p-q\rVert\leq 1$,
and 
(ii) there exists a minimizer $r\in\mathbb{G}^n$ of $h_{d_1,M}$ satisfying $\lVert q-r\rVert\leq 1$.

\vspace{5pt}
\noindent (i) Suppose not.
Let $q\in\mathbb{G}^n$ be the minimizer of $h_{d_{1/2},M}$ having the minimum distance from $p$.
Then $\lVert p-q\rVert\geq 2$ by the assumption.
We have $h_{d,M}(p\sqcap q), h_{d,M}(p\sqcup q) \leq h_{d,M}(q)=h_{d_{1/2},M}(q)$.
Indeed, if $h_{d,M}(p\sqcap q)>h_{d,M}(q)$,
then $h_{d,M}(p)+h_{d,M}(q)< h_{d,M}(p\sqcap q)+h_{d,M}(p \sqcup q)$ by the optimality of $p$; a contradiction to L-convexity of $h_{d,M}$.

The nonnegativity of $p\sqcap q$ and $p\sqcup q$ is clear.
Now we show that at least one of $p\sqcap q$ and $p\sqcup q$ satisfies \eqref{cond:pot_h} for $d_{1/2}$.
Then such one is a potential under the edge-cost $d_{1/2}$
(since it obviously satisfies \eqref{eq:dual_mnmf_d}).
Moreover, by Lemma~\ref{lem:d_property}, it is a minimizer of $h_{d_{1/2},M}$ with a smaller distance from $p$ than $q$; a contradiction.

First we observe that both $p\sqcap q$ and $p\sqcup q$ satisfy \eqref{cond:pot_h} for $e\neq e_0$ since both of them are potentials under $d$.
For $e=e_0=ij$, we have $\pi(p_i,p_j)=2d(ij)$ since $p$ is not a potential under $d_{1/2}$.
Also we have $\pi(q_i,q_j)\leq 2d(ij)-1$ by the definition.
Let $p':=(p+q)/2$.
By Lemma \ref{lem:property_h} (\ref{p_h_conv}), we obtain $\pi(p'_i,p'_j)\leq 2d(ij)-1/2$.
If $\pi(p'_i,p'_j)< 2d(ij)-1$, then $\pi((p\sqcap q)_i,(p\sqcap q)_j), \pi((p\sqcup q)_i,(p\sqcup q)_j)\leq 2d(ij)-1$
by Lemma \ref{lem:property_h} (\ref{p_h_int})(\ref{p_h_near}).
Thus both $p\sqcap q$ and $p\sqcup q$ satisfy \eqref{cond:pot_h} for $d_{1/2}$.
If $\pi(p'_i,p'_j)\geq 2d(ij)-1$, then $\pi(p'_i,p'_j)=2d(ij)-1/2$ or $2d(ij)-1$.
The condition of Lemma \ref{lem:property_h} (\ref{p_h_sym}) is satisfied since $d(ij)\geq 2$.
Hence at least one of $p\sqcap q$ and $p\sqcup q$ satisfies \eqref{cond:pot_h} for $d_{1/2}$.

\vspace{5pt}
\noindent (ii) The proof is almost the same as (i).
The difference lies in that $h_{d_{1/2},M}$ is not an L-convex function since $d_{1/2}$ is not integral.
But we can show that for a potential $r$ under the edge cost $d_1$, 
both $q\sqcap r$ and $q\sqcup r$ satisfy \eqref{cond:pot_h} for $d_{1/2}$.
Indeed, by Lemma \ref{lem:property_h} (\ref{p_h_conv})(\ref{p_h_near}) and \eqref{cond:pot_h} we have
\begin{align*}
\pi((q\sqcap r)_i,(q\sqcap r)_j)&\leq 1+\frac{(\pi(q_i,q_j)+\pi(r_i,r_j))}{2}\\
    &\leq 1+\frac{2d(ij)-1+2d(ij)-2}{2}\leq 2d(ij)-\frac{1}{2}
\end{align*}
for $e_0=ij$, and thus $\pi((q\sqcap r)_i,(q\sqcap r)_j)\leq 2d(ij)-1$ by Lemma \ref{lem:property_h} (\ref{p_h_int}).
Similarly, $\pi((q\sqcup r)_i,(q\sqcup r)_j)\leq 2d(ij)-1$.
For other edges, since both $q\sqcap r$ and $q\sqcup r$ are potentials under $d$, they satisfy \eqref{cond:pot_h} for $d_{1/2}$.
Hence we have $h_{d_{1/2},M}(q\sqcap r),h_{d_{1/2},M}(q\sqcup r)<\infty$.

Let $r\in\mathbb{G}^n$ be the minimizer of $h_{d_1,M}$ having the minimum distance from $q$.
Suppose that $\lVert q-r\rVert\geq 2$.
By the linearity of $\sum_{i\in V\setminus S}c_i y_i$,
we have $h_{d_{1/2},M}(q)+h_{d_{1/2},M}(r)=h_{d_{1/2},M}(q\sqcap r)+h_{d_{1/2},M}(q\sqcup r)$.
Then it holds that $h_{d_{1/2},M}(q\sqcap r),h_{d_{1/2},M}(q\sqcup r)\leq h_{d_{1/2},M}(r)$
by the optimality of $q$ under $d_{1/2}$.
We can show that at least one of $q\sqcap r$ and $q\sqcup r$ satisfies \eqref{cond:pot_h} for $d_1$; its proof is almost the same as (1).
Hence by Lemma \ref{lem:d_property} it is a minimizer of $h_{d_1,M}$
with a smaller distance from $q$ than $r$; a contradiction.
\end{proof}

\subsection{Finding a steepest descent direction (Proof of Lemma \ref{lem:sdd_algo})}
\label{subsec:sdd}

In this subsection, we show that a steepest descent direction at the current potential $p=(x,y)\in \mathbb{G}^n$
can be obtained from a maximum violating cut of an instance $((U,\tilde{E};\partial),\underline{c},\overline{c},\beta)$,
where $\beta$ is defined by \eqref{def:beta} (and \eqref{def:beta_c}\eqref{def:ex_beta_c}).
The proof is almost a straightforward extension of that given in \cite{Hirai2018Dual},
but is technically and notationally simplified.

For $(Y,Z)\in 3^U$ and $i\in V$, let $(Y,Z)_i:=(Y\cap U_i,Z\cap U_i)$.
We say that a cut $(Y,Z)\in 3^U$ is \emph{movable} if it satisfies the following condition:
\begin{itemize}
\item For $s\in S$, it holds that $(Y,Z)_s=(\emptyset,\emptyset)$.
\item For $i\in V\setminus S$ with $y_i=0$, it holds that $(Y,Z)_i$ is one of
\begin{equation}
\label{list:movable_1}
    (U_i,\emptyset),\ (\{i^t\},\emptyset),\ (\{i^t\},U_i\setminus\{i^t\}),\ (\emptyset,\emptyset),
\end{equation}
where $t\in\{1,\dotsc,k\}$ if $x_i=0$, and $t\in\{0,+\}$ if $x_i\neq 0$.
\item For $i\in V\setminus S$ with $y_i>0$, it holds that $(Y,Z)_i$ is one of \eqref{list:movable_1} and
\begin{equation}
\label{list:movable_2}
    (\emptyset,U_i\setminus\{i^t\}),\ (\emptyset,U_i),
\end{equation}
where $t\in\{1,\dotsc,k\}$ if $x_i=0$, and $t\in\{0,+\}$ if $x_i\neq 0$.
\end{itemize}

Recall that a violating cut $(Y,Z)\in 3^U$ is called maximum
if $\kappa(Y,Z)-\beta(Y,Z)$ is maximum among all violating cuts.

\begin{lemma}
\label{lem:proc}
From a maximum violating cut $(Y,Z)\in 3^U$,
we can obtain a movable maximum violating cut by the following procedure:
\begin{enumerate}
\renewcommand{\labelenumi}{{\rm (\Alph{enumi}).}}
\item For each $s\in S$, if $Z\cap U_s\neq \emptyset$,
	then replace $Z$ by $Z\setminus U_s$.
\item For each $i\in V\setminus S$ with $y_i=0$ and $Y\cap U_i=\emptyset$,
	if $Z\cap U_i\neq\emptyset$, then replace $Z$ by $Z\setminus U_i$.
\item For each $i\in V\setminus S$ with $x_i=0$ and $\lvert Y\cap U_i\rvert\leq 1$,
    if $1\leq\lvert Z\cap U_i\rvert\leq k-2$, then replace $Z$ by $Z\setminus U_i$.
\item For each $i\in V\setminus S$ with $x_i=0$ and $\lvert Y\cap U_i\rvert\geq 2$, replace $(Y,Z)$ by $(Y\cup U_i,Z\setminus U_i)$.
\end{enumerate}
\end{lemma}

\begin{proof}
Let $(Y,Z)\in 3^U$ be a maximum violating cut.
Then we have $\kappa(Y,Z)>\beta(Y,Z)>-\infty$.
If $Y\cap U_s\neq \emptyset$ for $s\in S$, then $\kappa(Y,Z)=-\infty$; a contradiction.
Thus $Y\cap U_s=\emptyset$.
If $Z\cap U_s\neq\emptyset$ for $s\in S$,
then the removal of $s^0$ from $Z$ does not change $\kappa(Y,Z)$ (and $\beta(Y,Z)$).
Hence after the procedure (A),
$(Y,Z)$ is still a maximum violating cut and $(Y,Z)_s=(\emptyset,\emptyset)$.

Consider a nonterminal node $i\in V\setminus S$ with $x_i\neq 0$.
If $y_i>0$, then any pair of distinct subsets of $U_i$ satisfies the condition of movability.
So we assume $y_i=0$.
If $i^0\in Z$ and $i^+\notin Y$, then the removal of $i^0$ from $Z$ does not change $\kappa(Y,Z)$.
Similarly, if $i^+\in Z$ and $i^0\notin Y$,
then the removal of $i^+$ from $Z$ does not change $\kappa(Y,Z)$.
Hence after the procedure (B),
$(Y,Z)$ is still a maximum violating cut
and $(Y,Z)_i$ is one of \eqref{list:movable_1}.

Consider a nonterminal node $i\in V\setminus S$ with $x_i=0$ and $y_i=0$.
If $Y\cap U_i=\emptyset$ and $Z\cap U_i\neq\emptyset$,
then the removal of $Z\cap U_i$ from $Z$
does not change $\kappa(Y,Z)$ and $\beta(Y,Z)$ by \eqref{def:beta_c}.
Similarly, if $1\leq \lvert Z\cap U_i\rvert\leq k-2$,
then the removal of $Z\cap U_i$ from $Z$
does not change $\kappa(Y,Z)$ and $\beta(Y,Z)$.
If $2\leq \lvert Y\cap U_i\rvert\leq k-1$,
then replacing $Y$ with $Y\cup U_i$ and $Z$ with $Z\setminus U_i$
does not change $\kappa(Y,Z)$ and $\beta(Y,Z)$.
Hence after the procedure (B)(C)(D),
$(Y,Z)$ is still a maximum violating cut
and $(Y,Z)_i$ is one of \eqref{list:movable_1}.

Consider a nonterminal node $i\in V\setminus S$ with $x_i=0$ and $y_i>0$.
As with the previous case,
if $1\leq \lvert Z\cap U_i\rvert\leq k-2$,
then the removal of $Z\cap U_i$ from $Z$ does not change
$\kappa(Y,Z)$ and $\overline{\beta}(Y,Z)$ by \eqref{def:ex_beta_c}.
If $\lvert Y\cap U_i\rvert \geq 2$ and $Z\cap U_i=\emptyset$,
then replacing $Y$ with $Y\cup U_i$ does not change
$\kappa(Y,Z)$ and $\overline{\beta}(Y,Z)$.
Hence after the procedure (C)(D),
$(Y,Z)$ is still a maximum violating cut
and $(Y,Z)_i$ is one of \eqref{list:movable_1} and \eqref{list:movable_2}.
\end{proof}

Next we show that there is a bijection between a subset of all movable cuts
and a subset of neighbors $\mathcal{F}_{p}\cup \mathcal{I}_{p}$ of the current potential $p=(x,y)$;
recall \eqref{eq:neighbor} in Section~\ref{subsec:l_conv}.
Then it turns out that the movable maximum violating cut gives rise to a steepest descent direction at $p=(x,y)$; see Figure~\ref{fig:neighbor}.

\begin{figure}[t]
\centering
\includegraphics[scale=0.9]{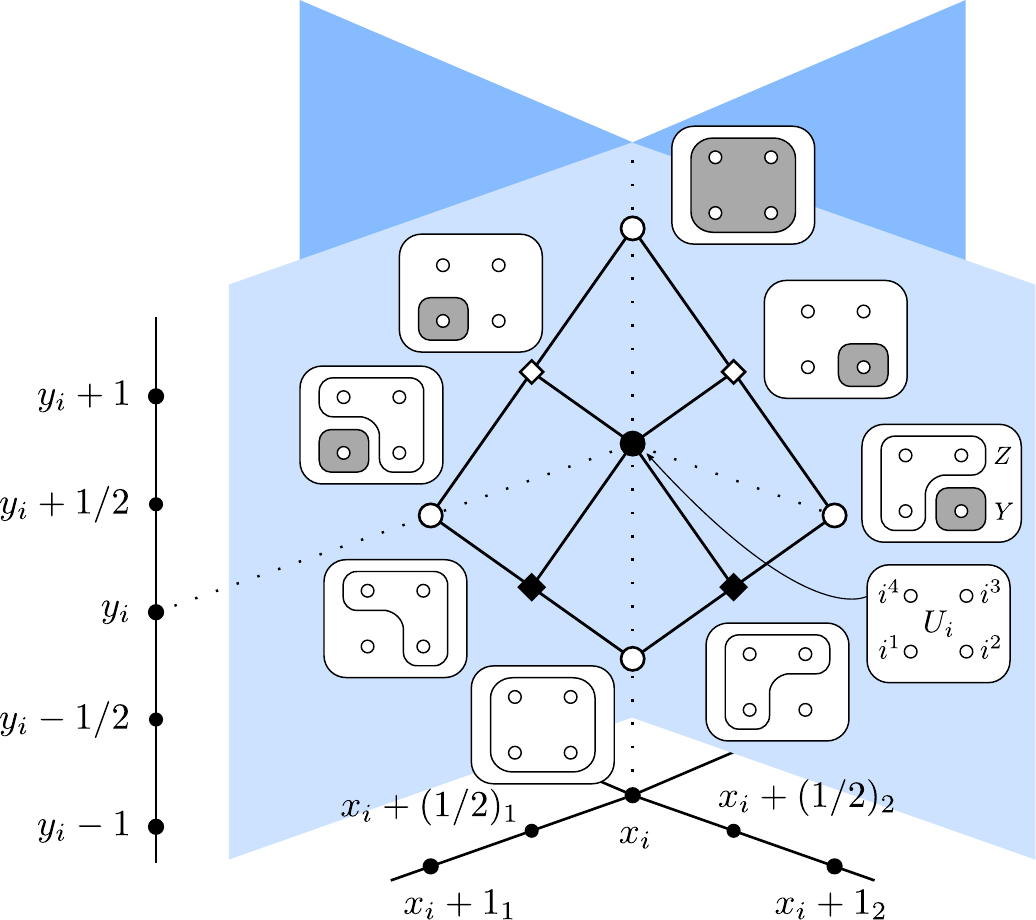}
\caption{The correspondence between $\mathcal{M}_{\mathcal{F},i}$ and $\mathcal{F}_i^+$.}
\label{fig:neighbor}
\end{figure}

For $x=(x,s)\in \mathbb{T}$, and $\varepsilon\in\{1/2, 1\}$, we use the following simplified notation.
\begin{align}
    x+\varepsilon_t&:=(x+\varepsilon, t)\quad \text{if $x\neq 0$ and $t=s$, or $x=0$,}\\
    x+\varepsilon_0&:=(x-\varepsilon, s)\quad \text{if $x\geq \varepsilon$.}
\end{align}
Let $p_i+(\varepsilon_t,\delta):=(x_i+\varepsilon_t,y_i+\delta)$.
We divide $\mathbb{G}:=\mathbb{G}_e\cup \mathbb{G}_p\cup \mathbb{G}'_e\cup \mathbb{G}'_o$ as
\begin{align}
    \mathbb{G}_e&=\{u\in \mathbb{G}\mid u:\text{even}\},\\
    \mathbb{G}_o&=\{u\in \mathbb{G}\mid u:\text{odd}\},\\
    \mathbb{G}'_e&=\{u\in \mathbb{G}\mid u:\text{non-integral, and}\ u+((1/2)_0,1/2):\text{even}\},\\
    \mathbb{G}'_o&=\{u\in \mathbb{G}\mid u:\text{non-integral, and}\ u+((1/2)_0,1/2):\text{odd}\}.
\end{align}
In Figure~\ref{fig:twisted_grid}, the points $\mathbb{G}'_e$ and $\mathbb{G}'_o$ are filled in black and white, respectively.
Then let $U_\mathcal{I}$ and $U_\mathcal{F}$ be the subsets of $U$ defined by
\begin{align}
U_\mathcal{I}&:=\{i^0 \mid p_i\in \mathbb{G}'_e\} \cup \{i^+ \mid p_i\in \mathbb{G}'_o\}\cup \bigcup_{p_i\in \mathbb{G}_o} U_i,\\
U_\mathcal{F}&:=\{i^+ \mid p_i\in \mathbb{G}'_e\} \cup \{i^0 \mid p_i\in \mathbb{G}'_o\}\cup \bigcup_{p_i\in \mathbb{G}_e} U_i.
\end{align}
A movable cut $(Y,Z)$ is said to be \emph{$\mathcal{F}$-movable} if $Y\cap U_\mathcal{I}=Z\cap U_\mathcal{I}=\emptyset$,
and be \emph{$\mathcal{I}$-movable} if $Y\cap U_\mathcal{F}=Z\cap U_\mathcal{F}=\emptyset$.
We define the set $\mathcal{M}_\mathcal{F},\mathcal{M}_\mathcal{I}\subseteq 3^U$ by
\begin{align}
    \mathcal{M}_\mathcal{F}&:=\{(Y,Z)\in 3^U\mid (Y,Z): \text{$\mathcal{F}$-movable cut}\},\\
    \mathcal{M}_\mathcal{I}&:=\{(Y,Z)\in 3^U\mid (Y,Z): \text{$\mathcal{I}$-movable cut}\}.
\end{align}
For $(Y,Z)\in \mathcal{M}_\mathcal{F}\cup\mathcal{M}_\mathcal{I}$, let
\begin{equation}
\label{eq:map_nonterminal}
p^{Y,Z}(i):=\begin{dcases*}
    p_i+((1/2)_t,1/2) & if $(Y,Z)_i=(\{i^t\},\emptyset)$, \\
    p_i+((1/2)_t,-1/2) & if $(Y,Z)_i=(\emptyset,U_i\setminus\{i^t\})$, \\
    p_i+(1_t,0) & if $(Y,Z)_i=(\{i^t\},U_i\setminus\{i^t\})$, \\
    p_i+(0,1) & if $(Y,Z)_i=(U_i,\emptyset)$, \\
    p_i+(0,-1) & if $(Y,Z)_i=(\emptyset,U_i)$, \\
    p_i & if $(Y,Z)_i=(\emptyset,\emptyset)$
\end{dcases*}
\end{equation}
for each nonterminal node $i\in V\setminus S$, and $p^{Y,Z}(s):=((M,s),0)$ for each terminal $s\in S$.
Recall that $(x',y')\in \mathbb{G}^n$ is said to be nonnegative if $y'_i\geq 0$ for all $i\in V$.
Let $\mathcal{F}^+_{p}:=\{q\in\mathcal{F}_{p}\mid q:\text{nonnegative}\}$
and $\mathcal{I}^+_{p}:=\{q\in\mathcal{I}_{p}\mid q:\text{nonnegative}\}$.

\begin{lemma}
\label{lem:bisubcut_representation}
\begin{enumerate}
\renewcommand{\labelenumi}{{\rm (\arabic{enumi}).}}
\item The map $(Y,Z)\mapsto p^{Y,Z}$ is a bijection from $\mathcal{M}_\mathcal{F}$ to $\mathcal{F}^+_{p}$
    and from $\mathcal{M}_\mathcal{I}$ to $\mathcal{I}^+_{p}$.
\item For $(Y,Z)\in \mathcal{M}_\mathcal{F}\cup\mathcal{M}_\mathcal{I}$, it holds
\begin{equation}
\label{eq:bisubcut_representation}
h(p^{Y,Z})-h(p)=-\frac{\kappa(Y,Z)-\beta(Y,Z)}{2},
\end{equation}
where $h$ is defined in \eqref{eq:h}.
\item Let $(Y_\mathcal{F},Z_\mathcal{F}):=(Y\cap U_\mathcal{F},Z\cap U_\mathcal{F})$
    and $(Y_\mathcal{I},Z_\mathcal{I}):=(Y\cap U_\mathcal{I},Z\cap U_\mathcal{I})$ for a movable cut $(Y,Z)\in 3^U$.
    Then $(Y_\mathcal{F},Z_\mathcal{F})$ is $\mathcal{F}$-movable, $(Y_\mathcal{I},Z_\mathcal{I})$ is $\mathcal{I}$-movable,
    and they satisfy
\begin{equation}
    \kappa(Y,Z)-\beta(Y,Z)
    =\kappa(Y_\mathcal{F},Z_\mathcal{F})-\beta(Y_\mathcal{F},Z_\mathcal{F})
    +\kappa(Y_\mathcal{I},Z_\mathcal{I})-\beta(Y_\mathcal{I},Z_\mathcal{I}).
\end{equation}
Conversely, for any $\mathcal{F}$-movable cut $(Y',Z')$ and $\mathcal{I}$-movable cut $(Y'',Z'')$,
the cut $(Y'\cup Y'',Z'\cup Z'')$ is a movable cut.
\end{enumerate}
\end{lemma}

We first prove Lemma~\ref{lem:sdd_algo} by using Lemma~\ref{lem:bisubcut_representation}.

\begin{proof}[Proof of Lemma~\ref{lem:sdd_algo}]
For a potential $p$, we construct the bidirected network $((U,\tilde{E};\partial),\underline{c},\overline{c})$ as in Section \ref{subsec:ext_bidir},
and define a reducible bisubmodular function $\beta$ by \eqref{def:beta}.
By the assumption, there is no feasible bisubmodular flow of this instance.
Then by Theorem~\ref{thm:bisubmod_flow} and \eqref{eq:nodeflowing_ineq1}--\eqref{eq:nodeflowing_ineq6},
we can obtain a maximum violating cut $(Y^*,Z^*)\in 3^U$ in $O(\mathrm{SF}(kn,m,k))$ time.
Applying the procedure of Lemma~\ref{lem:proc}, we obtain a movable maximum violating cut $(\tilde{Y},\tilde{Z})$.
Then by Lemma \ref{lem:bisubcut_representation},
a potential $p^{\tilde{Y}_\mathcal{F},\tilde{Z}_\mathcal{F}}$ is a minimizer of $h$ over $\mathcal{F}_{p}$,
and a potential $p^{\tilde{Y}_\mathcal{I},\tilde{Z}_\mathcal{I}}$ is a minimizer of $h$ over $\mathcal{I}_{p}$.
Thus we can obtain a steepest descent direction at $p$
by comparing $h(p^{\tilde{Y}_\mathcal{F},\tilde{Z}_\mathcal{F}})$ and $h(p^{\tilde{Y}_\mathcal{I},\tilde{Z}_\mathcal{I}})$.
\end{proof}

\begin{proof}[Proof of Lemma~\ref{lem:bisubcut_representation}]
(1) We claim that \eqref{eq:map_nonterminal} is a bijection from
$\mathcal{M}_{\mathcal{F},i}:=\{(Y,Z)_i\mid (Y,Z): \text{$\mathcal{F}$-movable}\}$
to $\mathcal{F}^+_{i}:=\{q\in \mathbb{G}\mid q\succeq p_i,\ q:\text{nonnegative}\}$.
Suppose that $p_i$ is odd, which implies $U_i\subseteq U_\mathcal{I}$.
Then it holds that $\mathcal{M}_{\mathcal{F},i}=\{(\emptyset,\emptyset)\}$ and $\mathcal{F}^+_{i}=\{p_i\}$.
Thus the claim follows.
Suppose that $p_i$ is even, which implies $U_i\subseteq U_\mathcal{F}$.
Then we can check the claim from Figure~\ref{fig:neighbor} for the case $y_i>0$.
For the case $y_i=0$, we can see the correspondence similarly.
Suppose that $p_i$ is non-integral. Note that in this case we have $y_i>0$.
If $p_i\in \mathbb{G}'_e$, we have $\mathcal{M}_{\mathcal{F},i}=\{(\emptyset,\emptyset),(\{i^+\},\emptyset),(\emptyset,\{i^+\})\}$,
and $\mathcal{F}^+_i=\{p_i,p_i+((1/2)_+,1/2),p_i+((1/2)_0,-1/2)\}$.
If $p_i\in \mathbb{G}'_o$, we have $\mathcal{M}_{\mathcal{F},i}=\{(\emptyset,\emptyset),(\{i^0\},\emptyset),(\emptyset,\{i^0\})\}$,
and $\mathcal{F}^+_i=\{p_i,p_i+((1/2)_0,1/2),p_i+((1/2)_+,-1/2)\}$.
Thus the claim follows.
We can similarly show that \eqref{eq:map_nonterminal} is a bijection between
$\mathcal{M}_{\mathcal{I},i}:=\{(Y,Z)_i\mid (Y,Z):\text{$\mathcal{I}$-movable}\}$
and $\mathcal{I}^+_i:=\{q\in\mathbb{G}\mid q\preceq p_i,\ q:\text{nonnegative}\}$.

\vspace{5pt}
\noindent (2) Let $p':=(x',y'):=p^{Y,Z}$.
First we show that $h(p')<\infty$ (i.e., $p'$ is a potential) if and only if $\kappa(Y,Z)>-\infty$.

Suppose that $h(p')<\infty$.
Consider an edge $e\in E_=$.
Then by $d_e>0$, there exists two nodes $i$ and $j$ with $x_j\neq 0$ such that $e=i^t j^0$
($t\in \{1,\dotsc,k\}$ if $x_i=0$, and $t\in \{0,+\}$ if $x_i\neq 0$).
We show that if one of $i^t$ and $j^0$ belongs to $Z$, then the other belongs to $Y$.
Suppose that $i^t\in Z$.
Recall that $\pi(p_i,p_j):=\mathrm{dist}(x_i,x_j)-y_i-y_j$.
By the definition of $E_=$,
we have $\pi(p_i,p_j) =2d_{ij}$.
By $i^t\in Z$ and \eqref{eq:map_nonterminal}, $x'_i$ goes away from $x_j$ or $y'_i$ decreases,
and hence $\pi(p'_i,p_j)=\pi(p_i,p_j)+1$.
Since $p'$ is also a potential, we must have $\pi(p'_i,p'_j)=\pi(p'_i,p_j)-1$.
Thus $p'_j$ must be $p_j+((1/2)_0,1/2)$, $p_j+(1_0,0)$, or $p_j+(0,1)$.
In all cases, $j^0$ belongs to $Y$.
In the case $j^0\in Z$, we can also prove $i^t\in Y$ similarly.
Now it follows that $\kappa(Y,Z)>-\infty$ from the above and $Y\cap U_s=Z\cap U_s=\emptyset$ for each $s\in S$.

Conversely, suppose that $\kappa(Y,Z)>-\infty$.
Since $p'$ is nonnegative, it suffices to show that $\pi(p'_i,p'_j)\leq 2d_{ij}$ for each edge $ij\in E$.
We assume that $ij\in E$ is replaced to $i^+j^0\in \tilde{E}$.
Since both replacements $p_i\rightarrow p'_i$ and $p_j\rightarrow p'_j$ change the value of $\pi(p_i,p_j)$ at most one,
we may consider the case of $\pi(p_i,p_j)=2d_{ij}$ or $2d_{ij}-1$ by the integrality (Lemma \ref{lem:property_h} (\ref{p_h_int})).
For the case $\pi(p_i,p_j)=2d_{ij}$, the argument is similar to the above.
Indeed, if $\pi(p'_i,p_j)-\pi(p_i,p_j)=1$, then $i^+\in Z$ and hence $j^0\in Y$.
So we have $\pi(p'_i,p'_j)-\pi(p'_i,p_j)=-1$.

For the case $\pi(p_i,p_j)=2d_{ij}-1$, it is sufficient to show that $\pi(p'_i,p'_j)\neq 2d_{ij}+1$.
Suppose not.
Then $p_i\neq p'_i$, $p_j\neq p'_j$.
Suppose that both $p_i$ and $p_j$ are integral.
By definitions, $\pi(p_i,p_j)$ and $\lVert p_i-p_j\rVert$ have the same parity (odd or even),
and thus $\lVert p_i-p_j\rVert$ is odd.
Then exactly one of $p_i$ and $p_j$ is odd and the other is even.
Hence at least one of $p_i=p'_i$ and $p_j=p'_j$ must hold; a contradiction.
Suppose that $p_i$ is non-integral and $p_j$ is integral.
Then it holds $p'_i=p_i+((1/2)_0,-1/2)$.
Since $\pi(p'_i,p_j)$ is even, by the same argument as above, $p'_i$ has the same parity as $p_j$.
This implies $p'_j=p_j$; a contradiction.
In the case where $p_i$ is integral and $p_j$ is non-integral, a contradiction can be derived similarly.
Suppose that both $p_i$ and $p_j$ are non-integral.
Then it holds that $p'_i=p_i+((1/2)_0,-1/2)$ and $p'_j=p_j+((1/2)_+,-1/2)$.
Also by $\pi(p'_i,p'_j)$ is odd, one of $p'_i$ and $p'_j$ is odd and the other is even; a contradiction.

For other cases, i.e., $ij\in E$ is replaced to $i^sj^0$ or $i^0j^0$, we can similarly show that $\pi(p'_i,p'_j)\leq 2d_{ij}$.
In fact, under the restriction of $\mathbb{G}$ in $((\mathbb{R}_+\times\{s'\})\cup (\mathbb{R}_+\times\{s\}))\times\mathbb{R}\subseteq \mathbb{T}\times\mathbb{R}$,
where $p_i=((x_i,s'),y_i)$ and $p_j=((x_j,s),y_j)$, we can say the above argument for these cases by a slight modification.
Thus we have $h(p')<\infty$.

Finally we show \eqref{eq:bisubcut_representation} for $(Y,Z)\in \mathcal{M}_\mathcal{F}\cup\mathcal{M}_\mathcal{I}$
with $h(p')<\infty$ and $\kappa(Y,Z)>-\infty$.
By the definition of $((U,\tilde{E};\partial),\underline{c},\overline{c})$,
we can check that twice the value of the right hand side of \eqref{eq:bisubcut_representation} is equal to
\begin{equation}
\sum_{i:x_i=0,y_i=0} \beta_{c_i}((Y,Z)_i)+\sum_{i:x_i=0,y_i>0} \overline{\beta}_{c_i}((Y,Z)_i)
+\sum_{i:x_i\neq 0} c_i(\lvert Y\cap U_i\rvert - \lvert Z\cap U_i \rvert).
\end{equation}
Indeed, for $i\in V$ with $x_i\neq 0$,
the case $\lvert Z\cap U_i\rvert > \lvert Y\cap U_i\rvert$ only occurs when $y_i>0$.
Thus $\underline{c}(i^0i^+)=c_i$ by the definition.
We can also check that each term is equal to $2c_i(y'_i-y_i)$ by the definitions.
Thus we have \eqref{eq:bisubcut_representation}.

\vspace{5pt}
\noindent (3) Let $(Y,Z)\in 3^U$ be a movable cut.
For a nonterminal node $i\in V\setminus S$ with $p_i$ is integral, $U_i$ is completely contained in $U_\mathcal{F}$ or $U_\mathcal{I}$.
We consider $i\in V\setminus S$ such that $p_i$ is non-integral.
Then it must hold $y_i>0$.
Thus any subset of $U_i$ satisfies the condition of movability.
Hence both $(Y_\mathcal{F},Z_\mathcal{F})$ and $(Y_\mathcal{I},Z_\mathcal{I})$ are movable.
The converse is obvious.

Finally, we show the equality.
It is clear that $\beta(Y,Z)=\beta(Y_\mathcal{F},Z_\mathcal{F})+\beta(Y_\mathcal{I},Z_\mathcal{I})$ by definition.
The contribution of edges in $E_-$ to the value of $\kappa(Y,Z)$ is
$-\sum_{i:x_i\neq 0}c_i(\lvert Y\cap U_i\rvert-\lvert Z\cap U_i\rvert)$; thus the equality is clear.
So it suffices to consider edges in $E_=$.
Fix an edge $i^+ j^0\in E_=$.
(For other cases, the similar argument holds.)
Since $\overline{c}(i^+j^0)=\infty$, we check that this capacity contributes to $\kappa(Y,Z)$
if and only if it contributes to $\kappa(Y_\mathcal{F},Z_\mathcal{F})$ or $\kappa(Y_\mathcal{I},Z_\mathcal{I})$.
We have $\mathrm{dist}(x_i,x_j)\geq \pi(p_i,p_j)=2d_{ij}\geq 2$ and $\pi(p_i,p_j)$ is even.
If $p_i$ and $p_j$ are both integral, then $p_i$ has the same parity (odd or even) as $p_j$
since the parities of $\pi(p_i,p_j)$ and $\lVert p_i-p_j\rVert$ are the same.
Thus $U_i\cup U_j$ is contained in $U_\mathcal{I}$ or $U_\mathcal{F}$.
If $p_i$ is non-integral and $p_j$ is integral,
then $p_i+((1/2)_0,1/2)$ has the same parity as $p_j$
since $\pi(p_i+((1/2)_0,1/2),p_j)=\pi(p_i,p_j)$ is even.
Thus $\{i^+\}\cup U_j$ is contained in $U_\mathcal{I}$ or $U_\mathcal{F}$.
If $p_i$ is integral and $p_j$ is non-integral,
we can similarly show that $U_i\cup \{j^0\}$ is contained in $U_\mathcal{I}$ or $U_\mathcal{F}$.
If $p_i$ and $p_j$ are both non-integral,
then both $p_i+((1/2)_0,1/2)$ and $p_j+((1/2)_+,1/2)$ have the same parity.
Thus $\{i^+,j^0\}$ is contained in $U_\mathcal{I}$ or $U_\mathcal{F}$.
In any cases, $\langle \partial(i^+j^0),\chi_{Y,Z}\rangle <0$ if and only if
$\langle \partial(i^+j^0),\chi_{Y_\mathcal{F},Z_\mathcal{F}}\rangle <0$ or
$\langle \partial(i^+j^0),\chi_{Y_\mathcal{I},Z_\mathcal{I}}\rangle <0$.
\end{proof}

\section*{Acknowledgements}
This work was partially supported by JSPS KAKENHI Grant Number JP17K00029.


\bibliography{mnmf.bib}

\end{document}